\documentclass[letterpaper,11pt]{article}
\usepackage[letterpaper,margin=1in]{geometry}
\usepackage{graphicx}
\usepackage{amsfonts}
\usepackage{amssymb}
\usepackage{amsmath,amsthm}
\usepackage{dsfont,mathtools}
\usepackage{relsize}
\usepackage{wrapfig}
\usepackage{frame,color}
\usepackage{hyperref}
\usepackage{xspace}
\usepackage{framed}
\usepackage{mathrsfs}
\usepackage{mathabx}
\usepackage{enumitem}
\usepackage{microtype}
\usepackage{doi}
\usepackage{hyperref}
\usepackage{thmtools}
\usepackage{footmisc}

\newtheorem{theorem}{Theorem}
\newtheorem{pretheorem}{Theorem}
\newtheorem{lemma}{Lemma}

\newtheorem{corollary}{Corollary}

\newcommand{\LOCAL}{\mathsf{LOCAL}}

\newcommand{\ignore}[1]{}

\newcommand{\logstar}{{\log^*}}

\newcommand{\ident}{{\operatorname{ID}}}

\newcommand{\ext}{\textrm{Ext}}
\newcommand{\extra}{\mathcal{EXT}}
\newcommand{\AAA}{\mathcal A}
\newcommand{\veee}{\mathop{\vee}}

\newcommand{\out}{\operatorname{out}}
\newcommand{\ing}{\operatorname{in}}
\newcommand{\none}{\operatorname{none}}
\newcommand{\II}{\mathcal I}
\newcommand{\JJ}{\mathcal J}

\newcommand{\GG}{\mathcal{G}}

\newcommand{\YY}{\mathcal{Y}}
\newcommand{\ZZ}{\mathcal{Z}}

\newcommand{\NN}{\mathbb{N}}
\newcommand{\OO}{\mathcal{O}}

\newenvironment{myabstract}
{\list{}{\listparindent 1.5em%
		\itemindent    \listparindent
		\leftmargin    1cm
		\rightmargin   1cm
		\parsep        0pt}%
	\item\relax}
{\endlist}

\newenvironment{mycover}
{\list{}{\listparindent 0pt
		\itemindent    \listparindent
		\leftmargin    1cm
		\rightmargin   1cm
		\parsep        0pt}%
	\raggedright
	\item\relax}
{\endlist}

\newcommand{\myemail}[1]{\,$\cdot$\, {\small #1}}
\newcommand{\myaff}[1]{\,$\cdot$\, {\small #1}\par\medskip}

\hypersetup{
	colorlinks=true,
	linkcolor=black,
	citecolor=black,
	filecolor=black,
	urlcolor=[rgb]{0,0.1,0.5},
	pdftitle={An Automatic Speedup Theorem for Distributed Problems},
	pdfauthor={Sebastian Brandt}
}

\begin{document}
	
	\begin{mycover}
		{\huge\bfseries\boldmath An Automatic Speedup Theorem for Distributed Problems \par}
		\bigskip
		\bigskip

		\textbf{Sebastian Brandt}
		\myemail{brandts@ethz.ch}
		\myaff{ETH Zurich}
		
	\end{mycover}
	
	\medskip
	\begin{myabstract}
		\noindent\textbf{Abstract.}
		Recently, Brandt et al.\ [STOC'16] proved a lower bound for the distributed Lov\'asz Local Lemma, which has been conjectured to be tight for sufficiently relaxed LLL criteria by Chang and Pettie [FOCS'17].
		At the heart of their result lies a speedup technique that, for graphs of girth at least $2t+2$, transforms any $t$-round algorithm for one specific LLL problem into a $(t-1)$-round algorithm for the same problem.
		We substantially improve on this technique by showing that such a speedup exists for \emph{any} locally checkable problem $\Pi$, with the difference that the problem $\Pi_1$ the inferred $(t-1)$-round algorithm solves is not (necessarily) the same problem as $\Pi$.
		Our speedup is \emph{automatic} in the sense that there is a fixed procedure that transforms a description for $\Pi$ into a description for $\Pi_1$ and \emph{reversible} in the sense that any $(t-1)$-round algorithm for $\Pi_1$ can be transformed into a $t$-round algorithm for $\Pi$.
		In particular, for any locally checkable problem $\Pi$ with exact deterministic time complexity $T(n, \Delta) \leq t$ on graphs with $n$ nodes, maximum node degree $\Delta$, and girth at least $2t+2$, there is a sequence of problems $\Pi_1, \Pi_2, \dots$ with time complexities $T(n, \Delta)-1, T(n, \Delta)-2, \dots$, that can be inferred from $\Pi$.
		
		As a first application of our generalized speedup, we solve a long-standing open problem of Naor and Stockmeyer [STOC'93]: we show that weak $2$-coloring in odd-degree graphs cannot be solved in $o(\logstar \Delta)$ rounds, thereby providing a matching lower bound to their upper bound. 
	\end{myabstract}

	\thispagestyle{empty}
	\setcounter{page}{0}
	\newpage

	\section{Introduction}\label{sec:introduction}

	In this work, we study the question of determining the time complexity of distributed graph problems from the perspective of \emph{round elimination}.
	More concretely, we ask: Given a problem $\Pi$, can we find another problem $\Pi'$ which can be solved in exactly one round less?
	Is it perhaps even possible to infer such a problem $\Pi'$ from $\Pi$ in an automated fashion, so that we can obtain a sequence of problems with decreasing complexities until we end up with a problem that can be solved in $0$ rounds?
	We will show that, given certain (reasonable) conditions, the answer to both questions is yes.
	As a concrete evidence of the power of this automatic speedup, we resolve the complexity of odd-degree weak $2$-coloring, an open question asked by Naor and Stockmeyer in 1993 \cite{DBLP:conf/stoc/NaorS93}.

	\paragraph{Model}
	Our distributed model of computation is a variant of the well-known $\LOCAL$ model \cite{DBLP:journals/siamcomp/Linial92, Peleg2000}, a synchronous message passing model where the nodes of a given input graph $G$ are processors that have the task to collaboratively solve some graph problem on $G$.
	The essential difference between our model and the $\LOCAL$ model is that in our case nodes are not equipped with unique identifiers, but (potentially) have some other symmetry breaking information available.
	The class of problems we consider are \emph{locally checkable} problems, i.e., problems where the global validity of a solution can be checked locally by the nodes in constant time.
	More precisely, for some constant $r$, each node has a set of acceptable output configurations for its radius-$r$ neighborhood, and a global solution is considered valid if and only if the output configuration of each node's radius-$r$ neighborhood is acceptable.
	Our speedup results are about the \emph{deterministic} time complexity of locally checkable problems; however, there exist known techniques to lift the obtained bounds to both the randomized and the deterministic $\LOCAL$ model.
	We will elaborate on these techniques in Section\ \ref{sec:lift}.

\subsection{Our Contributions}\label{sec:contrib}
	
	Our main contributions are twofold:
	1) We present a speedup theorem that generalizes the speedup technique of Brandt et al.\ \cite{Brandtlll}: In that work, the authors obtain a lower bound for the distributed Lov\'asz Local Lemma by developing a round elimination technique for a problem called sinkless orientation. We show that such a speedup exists for arbitrary locally checkable problems.
	In particular, we provide a method that takes the description of a problem $\Pi$ as input and outputs the description of a problem $\Pi'$ that can be solved one round faster but not any faster than that.
	This facilitates exciting new approaches for obtaining time complexity lower (and upper) bounds.
	2) We show that our speedup technique is a powerful tool by using it to prove a tight lower bound for weak $2$-coloring in odd-degree graphs, answering the long-standing open question by Naor and Stockmeyer.
	
	\paragraph{Speedup Results}
	We provide an automatic procedure that transforms any given locally checkable problem $\Pi$ into a (locally checkable) problem $\Pi_1$ such that, informally speaking, the following holds.
	\begin{pretheorem}[informal]\label{prethm:speedup}
		Let $\Pi$ be a locally checkable problem and $\Pi_1$ the problem obtained by applying our speedup transformation.
		On graphs of girth at least $2t+2$, the following two statements are equivalent:
		\begin{itemize}
			\item[(1)] There is an algorithm solving $\Pi$ in time $t$. 
			\item[(2)] There is an algorithm solving $\Pi_1$ in time $t-1$.
		\end{itemize}
	\end{pretheorem}

	If $\Pi$ has an exact deterministic time complexity $T(n, \Delta) \leq t$ on graphs with $n$ nodes, maximum node degree $\Delta$, and girth at least $2t+2$, then applying this speedup iteratively yields a sequence of problems $\Pi_1, \Pi_2, \dots$ with time complexities $T(n, \Delta)-1, T(n, \Delta)-2, \dots$.
	Now determining the time complexity of just a single $\Pi_i$ in the sequence will automatically determine the time complexities of all other $\Pi_j$, and, most importantly, of our initial problem $\Pi$.
	We will give a detailed explanation how to apply this technique to infer bounds for the complexity of a problem in Section \ref{sec:apply}, but let us first consider a concrete application.
	We remark that, apart from the following application, our speedup theorem also semi-automatically reproduces previously known techniques, such as the sinkless orientation speedup \cite{Brandtlll} and color reduction on rings (leading to the $O(\log^* n)$ upper bound for $3$-coloring a ring \cite{DBLP:conf/stoc/ColeV86, DBLP:conf/stoc/GoldbergPS87}), as we will see in Sections \ref{sec:sinkless} and\ \ref{sec:colred}.
	Moreover, building on our speedup technique, Balliu et al.\ \cite{DBLP:journals/corr/abs-1901-02441} very recently proved new lower bounds for maximal matching and maximal independent set, in both the deterministic and the randomized setting.

	\paragraph{Odd-Degree Weak $2$-Coloring}
	Weak $k$-coloring is the problem of coloring the nodes of a given input graph with $k$ colors such that each node (of degree at least $1$) has \emph{at least one} neighbor with a different color.
	In their seminal work \cite{DBLP:conf/stoc/NaorS93}, Naor and Stockmeyer proved that in graphs where each node has odd degree, a weak $2$-coloring can be found in time $O(\logstar \Delta)$.
	As one of their three open questions they asked whether this bound can be improved.
	While the question seemed simple enough, surprisingly no progress has been made over the past 25 years, indicating that the available lower bound techniques might not be sufficient for showing that the upper bound is tight.
	Using our speedup results, we close this gap by showing the following theorem.
	\setcounter{pretheorem}{3}
	\begin{pretheorem}\label{prethm:lsd}
		There is no $o(\logstar \Delta)$-time algorithm solving weak $2$-coloring in odd-degree graphs.
	\end{pretheorem}

	Given the speedup framework, the main technical ingredient in our proof is a generalization of weak $2$-coloring to a problem we call \emph{superweak $k$-coloring} that has the following two nice properties:
	\begin{enumerate}
		\item If we set $\Pi$ to be superweak $k$-coloring, then the problem $\Pi_1$ obtained by applying our speedup is at least as hard as superweak $k'$-coloring, for some $k' > k$.
		\item Relaxing $\Pi_1$ to superweak $k'$-coloring (and then continuing to apply the speedup technique) is sufficiently \emph{tight}, in the sense that we essentially still need $\Omega(\logstar \Delta)$ speedup steps until we obtain a $0$-round solvable problem. 
	\end{enumerate}
	More concretely, a rough proof outline goes as follows.
	Relax weak $2$-coloring to superweak $2$-coloring, apply our speedup, relax the obtained problem to superweak $k$-coloring for some $k>2$, apply our speedup, relax to superweak $k'$-coloring for some $k'>k$, etc.
	Then, show that any problem obtained after $o(\logstar \Delta)$ steps of speedup and relaxation is still not solvable in $0$ rounds.
	Our speedup results then immediately imply that there is no $o(\logstar \Delta)$-algorithm for weak $2$-coloring.

\subsection{Related Work}\label{sec:related}

	\paragraph{Follow-up Work}
	In a recent breakthrough, Balliu et al.\ \cite{DBLP:journals/corr/abs-1901-02441} used our speedup technique to show that, both for maximal matching and maximal independent set, there is no randomized algorithm with runtime $o(\Delta + \log \log n / \log \log \log n)$ and no deterministic algorithm with runtime $o(\Delta + \log n / \log \log n)$.
	As documented in \cite[Section 3.7]{DBLP:journals/corr/abs-1901-02441}, apart from the speedup provided by Theorem \ref{prethm:speedup}, the authors also apply both of our simplification techniques (see Section \ref{sec:apply}) to achieve their lower bounds.

	\paragraph{Further Related Work}
	The first occurrence of the round elimination speedup technique we extend to any locally checkable problem was seen in \cite{Brandtlll}, where the authors show that such a speedup works for the problem of ($\Delta$-regular) sinkless orientation, resulting in a randomized lower bound of $\Omega(\log \log n)$ that also applies to the (constructive symmetric) distributed Lov\'asz Local Lemma (LLL) and $\Delta$-coloring.
	In \cite{chang16separation}, Chang et al.\ proved that this result can be extended to a deterministic lower bound of $\Omega(\log_{\Delta} n)$, which is tight (for sinkless orientation) due to a matching upper bound by Ghaffari and Su \cite{DBLP:conf/soda/GhaffariS17}, who also gave a matching upper bound for the deterministic case.
	Chang and Pettie \cite{DBLP:conf/focs/ChangP17} conjectured that the (randomized) lower bound for the distributed LLL is tight for sufficiently relaxed LLL criteria; despite a recent improvement of the upper bound by Ghaffari et al.\ \cite{DBLP:conf/focs/GhaffariHK18}, this conjecture is still open.
	In \cite{DBLP:conf/soda/ChangHLPU18}, Chang et al.\ simplified the randomized speedup technique of \cite{Brandtlll} and showed that the sinkless orientation lower bounds also imply an $\Omega(\log_{\Delta} n)$ deterministic and an $\Omega(\log_{\Delta} \log n)$ randomized lower bound for $(2\Delta - 2)$-edge coloring.

	Weak $k$-coloring on odd-degree graphs was introduced by Naor and Stockmeyer \cite{DBLP:conf/stoc/NaorS93} as an example of a non-trivial problem that can be solved in constant time \emph{on graphs of bounded degree}.
	As they show, the odd-degree condition is a necessary requirement; there are graph classes where nodes are allowed to have even degrees for which no constant-time weak coloring algorithm exists.
	Very recently, Balliu et al.\ \cite{DBLP:journals/corr/abs-1811-01643} refined our knowledge in this regard by proving a tight lower bound of $\Omega(\log^* n)$ for weak $2$-coloring on regular trees.
	While, from a theory perspective, weak coloring is interesting as ``a problem with minimal symmetry breaking requirements" \cite[p.139]{DBLP:conf/spaa/Kuhn09}, a more concrete application exists in the form of certain resource allocation problems \cite{DBLP:conf/stoc/NaorS93}.
	Naor and Stockmeyer provided an $O(\logstar \Delta)$-algorithm for odd-degree weak $2$-coloring, which was subsequently simplified and adapted to the dynamic setting by Mayer et al.\ \cite{DBLP:conf/istcs/MayerNS95}.
	While progress has been made for other relaxations of the standard node coloring problem, such as defective \cite{DBLP:conf/spaa/Kuhn09, DBLP:conf/stoc/BarenboimE09} or arbdefective \cite{DBLP:conf/podc/BarenboimE10} coloring, the question by Naor and Stockmeyer whether their bound can be improved has remained open until now.

	Naor and Stockmeyer's work also initiated, together with \cite{DBLP:journals/siamcomp/Linial92}, the line of research on local algorithms, and introduced the concept of \emph{locally checkable labeling} (LCL) problems.
	This class of problems has been subject to many investigations, resulting in an almost complete understanding of the respective complexity landscape very recently \cite{DBLP:conf/wdag/BalliuBOS18, DBLP:conf/stoc/BalliuHKLOS18, Brandtlll, DBLP:conf/podc/BrandtHKLOPRSU17, chang16separation, DBLP:conf/focs/ChangP17, DBLP:conf/wdag/FischerG17, DBLP:conf/focs/GhaffariHK18}.
	We remark that, while our speedup applies to LCL problems, it is not restricted to them: in particular, we do not require the considered graphs to be of constant degree.

\section{Technical Overview}
	In this section, we will outline how to apply the speedup technique to obtain new bounds, including further helpful techniques for the application of the speedup.
	Moreover, we will give an overview of the available techniques to extend bounds achieved from our speedup (in our model) to the LOCAL model and to randomized complexities.

\subsection{How to Apply the Speedup Results}\label{sec:apply}
	The most natural application of our speedup technique is to prove a lower bound for some given problem $\Pi$.
	The roadmap is as follows.
	Starting with $\Pi$, we apply our speedup theorem iteratively, resulting in a problem sequence $\Pi, \Pi_1, \Pi_2, \dots$, where each problem can be solved exactly one round faster than the previous one.
	Given the mentioned conditions, our speedup works until we reach a problem that can be solved in $0$ rounds; hence, in theory, the only thing we have to do is to look at our sequence $\Pi, \Pi_1, \Pi_2, \dots$ and to determine which is the first problem in this sequence that is solvable in $0$ rounds (usually depending on our parameters $n$ and/or $\Delta$).
	If $\Pi_t$ is the first problem solvable in $0$ rounds, the problem $\Pi$ we are interested in has time complexity $t$.
	However, since Theorem \ref{prethm:speedup} requires girth at least $2t+2$, only the \emph{lower} bound of $t$ holds for general graphs.

	While $0$-round solvable problems have a simple characterization, there is a catch: in general, the description of an inferred problem $\Pi_i$ is much more complex than the description of the original problem.
	In fact, dealing with this explosion in complexity is one of the main challenges in applying our speedup.
	To this end, we provide two simplification techniques. 

	\paragraph{Relaxation}
	After inferring a new problem $\Pi_{i+1}$ from $\Pi_i$ via the speedup, we can try to find a relaxed version of $\Pi_{i+1}$ (i.e., a problem that is provably not harder than $\Pi_{i+1}$) that has a much simpler description, and use this problem as the starting point for the next speedup step.
	Alternating between relaxation and speedup, we continue this process until we reach a problem $\Pi_t$ that is solvable in $0$ rounds.
	Then $t$ is a \emph{lower bound} for the time complexity of our initial problem $\Pi$. 
	Of course, we can also stop before we reach a $0$-round solvable problem, and the respective index of the problem is also a lower bound.

	If, informally speaking, we relax the problems obtained after each speedup step too much, the lower bound we obtain in the end might be asymptotically worse than the correct (tight) bound, or no improvement on existing bounds at all, so finding the right relaxation is a challenging problem.
	Moreover, in order to avoid having to find ``good" relaxations for many very different problems, it is desirable to find relaxed problems that are \emph{similar} (perhaps with different parameters of some kind) to \emph{previous} problems in the problem sequence.
	The above outline captures exactly what we do in our lower bound for weak $2$-coloring.
	
	We remark that a dual version of the relaxation technique exists, where we obtain upper bounds on high-girth graphs by making problems harder instead of relaxing them.
	We will see a concrete example for this dual technique when we consider color reduction on rings as a special case of our speedup in Section\ \ref{sec:colred}.
	
	\paragraph{Description Simplification} As the second tool in our toolbox for managing the increasing description complexity, we provide a ``maximality constraint" that can be applied twice per speedup step in order to decrease the set of allowed outputs and thereby simplify the problem.
	Despite its simplicity, this technique can significantly reduce the description complexity of a problem, as we will see in Section\ \ref{sec:app2col}.
	As we show in Theorem\ \ref{thm:primespeedup}, this simplification comes at no cost, keeping the complexity of the problem under consideration as it is.

\subsection{How to Lift Bounds to the $\LOCAL$ Model.}\label{sec:lift}
	A requirement for our speedup result is that the class of input graphs satisfies a property that we call $t$-independence.
	Informally, a graph class is $t$-independent if the following holds: If any node\footnote{If we want to be a bit more precise, the same also has to hold for any edge, where for simplicity, we also consider edges as computational entities.} $v$ that has gathered all information in its radius-$(t-1)$ neighborhood extends its view by one hop along some edge, then the new information $v$ obtains does not affect what information $v$ might see if it extends its view by one hop along any other edge.
	In particular, if the nodes are equipped with globally unique identifiers, then $t$-independence does not hold: if a node sees some ID in the extended view along some edge, it knows that this ID cannot be in any of the extended views along the other edges, due to our girth condition.
	While almost every other kind of symmetry breaking information commonly used, such as node colorings, edge colorings, edge orientations, or combinations thereof, satisfy $t$-independence, extending bounds obtained by our speedup technique to the setting with unique IDs, i.e., the $\LOCAL$ model, requires additional techniques.
	Note that \emph{upper} bounds obtained by our technique immediately apply to the $\LOCAL$ model (as basically every other symmetry breaking information can be inferred from unique IDs), hence we will focus on lower bounds in the following.
	
	\paragraph{Method I: Randomization}
	As demonstrated in \cite{DBLP:journals/corr/abs-1901-02441, Brandtlll}, by explicitly incorporating error probabilities into the speedup steps, lower bounds in our setting can be lifted to the randomized $\LOCAL$ model (which essentially guarantees $t$-independence since no unique IDs are required).
	The obtained bounds are weaker than the original bounds from our setting, which is to be expected considering that allowing randomization can only lower the complexity of a problem.
	In a second step, the randomized bounds can then be lifted to the (deterministic) $\LOCAL$ model, by exploiting gaps in the complexity landscape of so-called LCL problems \cite{chang16separation}, or by explicitly showing that the existence of a deterministic algorithm of some complexity would imply the existence of a randomized algorithm that violates the randomized lower bound \cite{DBLP:journals/corr/abs-1901-02441}.
	The available evidence \cite{DBLP:journals/corr/abs-1901-02441, chang16separation} suggests that this detour via the randomized complexity does not weaken the deterministic bound: in both cases, the bound in our setting is identical to the bound in the $\LOCAL$ model.
	This is not too surprising since the uniqueness of IDs might simply not be enough to change the complexity of a problem (as compared to, say, a setting with non-unique IDs); however, a proof for this is not known and would be a valuable step forward. 

	\paragraph{Method II: Order-Invariant Algorithms}
	A second technique to lift bounds to the $\LOCAL$ model comes into play when we are interested in time complexities as a function of the maximum node degree $\Delta$ of the input graph.
	By a result of Naor and Stockmeyer \cite{DBLP:conf/stoc/NaorS93}, if there is a constant-time algorithm solving a locally checkable problem in the $\LOCAL$ model, then there is also an order-invariant algorithm with the same runtime, where order-invariant indicates that any node only uses the \emph{relative} IDs, i.e., the order of the IDs it sees, but not the actual ID values, in order to determine its output.
	Hence, if there is an algorithm solving a locally checkable problem with runtime independent of $n$, then we can essentially restrict attention to order-invariant algorithms.
	We provide an extension of our speedup result (Theorem\ \ref{thm:ordinv}) that shows that for order-invariant algorithms our speedup holds also in the case of unique IDs, i.e., in the $\LOCAL$ model.
	We will make use of this extension when we prove our lower bound for weak $2$-coloring (in the $\LOCAL$ model).

	\section{Preliminaries}\label{sec:prel}

	\paragraph{Graphs}
	All graphs we consider throughout the paper will be simple, undirected and connected.
	We denote the set of nodes of a graph $G$ by $V(G)$ and the set of edges by $E(G)$, and we set $n := |V(G)|$.
	Furthermore, we write $d(v)$ for the degree of a node $v$ and denote the maximum node degree of a graph by $\Delta$.
	A \emph{($\Delta$-)regular} graph is a graph where $d(v) = \Delta$ for all $v \in V(G)$.
	The \emph{girth} of a graph is the length of the smallest cycle.
	A \emph{matching} is a set $M \subseteq E(G)$ such that no two distinct edges from $M$ share an endpoint. 

	An important component in designing and proving our speedup is the idea to split the output of a node into parts that \emph{belong} to incident edges.
	As the basis for a convenient representation of such a split output, define $B(G)$ as the set of all pairs $(v, e)$ where $e \in E(G)$ and $v \in V(G)$ is an endpoint of $e$.
	Finally, for a graph class $\GG$, we denote the subclass of $\GG$ consisting of the contained graphs with $n$ nodes and maximum degree $\Delta$ by $\GG_{n, \Delta}$, for all non-negative integers $n, \Delta$.
	Similarly, the subclass of graphs with maximum degree $\Delta$ is denoted by $\GG_{\Delta}$.

	\paragraph{Input-Labeled Graphs}
	Commonly, locally checkable problems are defined by allowed configurations of output labels and a specification of the given inputs and the considered graph class.
	Definitionwise, we will strictly separate between the outputs on one side (which will define what we call a problem) and the inputs\footnote{Note that in this work, we focus on the common case of problems where the \emph{correctness} of the output does not depend on the given inputs. Adapting the speedup results to the case where output correctness depends on the inputs is not hard, but carries a significant technical overhead which would needlessly impair readability.} and the graph class on the other side (which will be given by input-labeled graphs).
	This enables us to give a very general definition of the setting in which our results are applicable; more importantly though, this separation caters to the fact that the sequence of problems we obtain by repeatedly speeding up a given problem is \emph{independent} of the considered inputs and the considered graph class.

	For the definition of input-labeled graphs, we will also use the set $B(G)$ defined above, which allows for a convenient way to encode, e.g., edge orientations.
	Furthermore, we will not restrict attention to graphs with bounded degree or to bounded input label sets; instead we will use the following more complicated, but also more general definition, which allows, e.g., to define graph classes of unbounded degree with an input edge coloring (which requires $\Omega(\Delta)$ labels), or unique IDs (which come from a set that is a function of $n$).
	
	Let $\Sigma$ be a (possibly infinite) set of input labels and $\iota : \NN^2 \rightarrow 2^{\Sigma}$ a function such that $\Sigma_{n, \Delta} :=\iota(n, \Delta)$ is a \emph{finite} subset\footnote{Note that, throughout the paper, we use the expression $2^S$ for the power set of set $S$, as opposed to the set of functions from $S$ to $\{ 0, 1\}$.} of $\Sigma$, for all $(n, \Delta) \in \NN^2$.
	A {$\Sigma$-input-labeled graph} is a pair $(G, \varphi_G)$, where $G$ is a graph and $\varphi_G$ is a function $\varphi_G : B(G) \rightarrow \Sigma_{n, \Delta}$, where $n$ and $\Delta$ are the number of nodes and the maximum degree of $G$, respectively.
	For simplicity, we will usually omit the function $\varphi_G$ and simply denote the $\Sigma$-input-labeled graph by $G$.
	We extend the notion of being $\Sigma$-input-labeled to graph classes and say that a graph class $\GG$ is $\Sigma$-input-labeled if each graph in $\GG$ is $\Sigma$-input-labeled.
	Throughout the paper, all considered graph classes are assumed to be input-labeled, if not stated otherwise.
	Furthermore, if we want to avoid that nodes have to be able to compute uncomputable functions during the distributed computation, we can require additionally that the function that maps each pair $(n, \Delta)$ to the graph class $\GG_{n, \Delta}$ (as well as any other function involved in specifying parts of a distributed problem) is computable.

	\paragraph{Problems}
	The speedup results we present apply to all locally checkable problems; however, formally, we will only consider problems where the validity of a global output essentially\footnote{Due to our particular way of splitting the output of a node into partial outputs for each incident edge, our formal definition of a problem will contain acceptable configurations for both edges and nodes; a more precise term for edge-checkability would thus be node-and-edge-checkability.\label{nodeedgefoot}} can be checked on edges, i.e., there is a set of acceptable output configurations for the two endpoints of an edge, and the global output is correct if and only if the configuration for each edge is acceptable.
	Restriction to these problems does not lose generality for our purposes: by requiring that each node outputs the computed output labels (and the topology\footnote{Note that in general graphs where nodes do not have unique identifiers, the information a node obtains in $t$ rounds may not be enough to determine the exact topology of the subgraph induced by all nodes in distance at most $t$; however, since we will only consider radius-$t$ neighborhoods in graphs that have girth at least $2t+2$, the subgraph topology in each radius-$t$ neighborhood is a tree which implies that each node can determine the exact topology of its radius-$t$ neighborhood.\label{repfoot}}) of its whole radius-$r$ neighborhood for some suitably large constant $r$, any locally checkable problem can be transformed into an edge-checkable\footref{nodeedgefoot} problem with the same asymptotic time complexity.
	For the definition of (our restricted version of) a problem, we will need the notion of a \emph{multiset}, which is simply a set in which elements can have \emph{multiplicity} larger than $1$, but where, as usual, the order of elements does not matter.

	Formally, for the scope of this paper, a problem $\Pi$ is defined by
	\begin{enumerate}
		\item a (possibly infinite) set $\mathcal O$ of output labels,
		\item a function $f : \NN \rightarrow 2^{\mathcal O}$ such that $f(\Delta)$ is a \emph{finite} subset of $\mathcal O$, for all $\Delta \in \NN$,
		\item a function $g$ that maps each $\Delta \in \NN$ to a set of $2$-element multisets where both elements are taken from $f(\Delta)$, and
		\item a function $h$ that maps each $\Delta \in \NN$ to a set of multisets with at most $\Delta$ elements, all taken from $f(\Delta)$ .
	\end{enumerate}
	
	\noindent The specification of $f$ ensures that problems that require the set of output labels to depend\footnote{Similarly to the case of unique IDs as input labels, one could also allow for problems where the set of output labels depends on $n$. However, we are not aware of commonly studied problems of this kind, which is why we chose the simpler definition that only allows a dependence on $\Delta$.} on $\Delta$, such as $(\Delta + 1)$-coloring, are included in our problem definition.
	The sets $g(\Delta)$ and $h(\Delta)$ formalize which output configurations are allowed on an edge $e=\{u, v\}$ (i.e., at $(u,e) \in B(G)$ and $(v,e) \in B(G)$), resp.\ at a node $v$ (i.e., at $(v, e_1), \dots, (v, e_{d(v)})$, where $e_1, \dots, e_{d(v)}$ are the edges incident to $v$).
	For instance, the problem of $(\Delta + 1)$-coloring can be described by setting $\mathcal O := \NN^+$, $f(\Delta) := \{ 1, \dots, \Delta\}$, $g(\Delta) := \{ \{ c_1, c_2 \} \mid c_1, c_2 \in f(\Delta), c_1 \neq c_2 \}$, $h(\Delta) := \{ \{ c_1, \dots, c_i \} \mid 0 \leq i \leq \Delta, c_1 = \dots = c_i \}$.

	Combining problems and input-labeled graphs, we define a \emph{realized problem} as a pair $(\Pi, \GG)$, where $\Pi$ is a problem and $\GG$ a $\Sigma$-input-labeled graph class.
	For convenience, we may simply use the term ``problem" for $(\Pi, \GG)$.
	We say that an algorithm $\AAA$ solves a realized problem $(\Pi, \GG)$ (or, equivalently, that $\AAA$ solves $\Pi$ on $\GG$) if, for any $\Delta$ and any graph $G \in \GG_{\Delta}$, $\AAA$ assigns an output $o_{v,e} \in f(\Delta)$ to each pair $(v,e) \in B(G)$ such that, for each edge $e = \{u, v\} \in E(G)$, the multiset $\{o_{u,e}, o_{v,e}\}$ is contained in $g(\Delta)$, and for each node $v \in V(G)$, the multiset $\{ o_{v, e_1}, \dots, o_{v,e_{d(v)}} \}$ is contained in $h(\Delta)$.

	\paragraph{Model}
	Since we defined inputs to be part of the considered graph class (and hence assigned the duty of providing sufficient symmetry-breaking information to the choice of the graph class), we can use a very weak model of computation.
	This has the advantage that essentially all problems that are defined in a stronger model, such as the $\LOCAL$ model, can also be formulated in our model.
	The only requirement that we need is that nodes are able to distinguish between their neighbors (or incident edges), which is why we formally choose the port numbering model \cite{DBLP:conf/stoc/Angluin80} as our model of computation.
	
	In the port numbering model, each node $v$ of the input graph $G \in \GG$ has $d(v)$ many ports $1, \dots, d(v)$ which correspond to the edges incident to $v$; the two endpoints of an edge may have different ports corresponding to the connecting edge.
	Each node can communicate with its neighbors by sending messages along the connecting edges.
	Computation proceeds in synchronous rounds where in each round each node first sends arbitrarily large messages to its neighbors and then, upon receiving the messages sent by its neighbors, performs some arbitrarily complex local computation.
	Each node executes the same algorithm and has to terminate at some point, upon which it outputs its local part of the global solution (e.g., if the task is to find a proper node coloring, each node outputs its own color).
	The runtime of such a distributed algorithm is the number of rounds until the last node terminates.
	In the beginning of the computation, each node $v$ is aware of the parameters $n$ and $\Delta$ and sees the input label assigned to $(v,e) \in B(G)$ for all incident edges $e$, i.e., one input label per port.
	When terminating, a node $v$ assigns an output label from $f(\Delta)$ to each $(v,e)$.
	We are interested in the \emph{worst-case} runtime of a distributed algorithm, i.e., for worst-case input-labeled graphs with worst-case assignments of port numbers to edges.

	A distributed algorithm is correct if the output labels assigned to the elements in $B(G)$ satisfy the constraints encoded in $g(\Delta)$ and $h(\Delta)$ at each edge, resp.\ node.
	More precisely, we say that a distributed algorithm solves a realized problem $(\Pi, \GG)$ if the above is true for each graph $G \in \GG$ (where $f(\Delta), g(\Delta), h(\Delta)$ are defined by $\Pi$).
	In this work, we formally only consider deterministic algorithms; the implications for randomized algorithms have been discussed in Section\ \ref{sec:lift}.
	
	It is well-known that the best a node can do in $t$ rounds of communication is to gather the whole input information contained in its radius-$t$ neighborhood (as well as the topology\footref{repfoot}), and then decide on its output using only the collected information.
	Hence, a $T(n, \Delta)$-round algorithm can be equivalently described as a function that maps each possible radius-$T(n, \Delta)$ neighborhood of a node $v$ to a tuple of $d(v)$ output labels (one for each port of $v$, i.e., each $(v,e) \in B(G)$).
	Note that in this description the degree of $v$ is only fixed when a radius-$T(n, \Delta)$ neighborhood has been chosen.

	\paragraph{Neighborhoods}
	Let $G$ be a $\Sigma$-input-labeled graph.
	For any node $v \in V(G)$ and any non-negative integer $t$, we define the \emph{radius-$t$ neighborhood} $N^t(v)$ of $v$ (in $G$) as the collection of information that $v$ can obtain in $t$ rounds, i.e., the topology of the subgraph of $G$ induced by the set of all nodes in distance at most $t$ from $v$, together with the input information each of these nodes possesses at the very beginning.\footref{repfoot}
	Similarly, for any edge $e = \{ u, v \}$, we define the \emph{radius-$t$ neighborhood} $N^t(e)$ of $e$ (in $G$) as the collection of information that \emph{both} $u$ and $v$ can obtain in $t$ rounds, i.e., the topology of the subgraph of $G$ induced by the set of all nodes in distance at most $t$ from both $u$ and $v$, together with the respective input information.
	For convenience, we may occasionally forget about assigned labels or edges and consider $N^t(v)$ or $N^t(e)$ as unlabeled graphs or sets of nodes.
	We write $N^t_G(v)$ and $N^t_G(e)$ if we want to specify the underlying graph $G$.
	We say that the radius-$t$ neighborhood of a node $v$ in graph $G$ and the radius-$t$ neighborhood of a node $u$ in graph $H$ are \emph{isomorphic} and write $N^t_G(v) \cong N^t_H(u)$ if the collection of information $v$ can obtain in $t$ rounds in $G$ is the same as the collection of information $u$ can obtain in $t$ rounds in $H$.
	We define isomorphisms for radius-$t$ edge neighborhoods analogously.

	For simplicity, we will abuse notation, and use set operations to describe further collections of information, in the canonical way.
	For instance, we write $N^t(e) = N^t(u) \cap N^t(v)$.
	In particular, we are interested in ``extensions" of such information collections, e.g., the information a node $v$ can obtain in $t$ rounds that a neighbor $u$ cannot obtain in $t$ rounds, or the information two neighboring nodes $u$ and $v$ can both obtain in $t$ rounds that $v$ cannot obtain in $t-1$ rounds.
	Hence, for any node $v$, any edge $e = \{ u, v \}$, and any positive integer $t$, we define $\ext^t_e(v) := N^t(e) \setminus N^{t-1}(v)$ and $\ext^t_v(e) :=  N^{t}(v) \setminus N^{t}(e)$.
	Again, we write $\ext^t_{e,G}(v)$ and $\ext^t_{v,G}(e)$ to specify the underlying graph.

	Consider some radius-$t$ neighborhood $N^t(e)$ of an edge $e = \{ u, v \} \in E(G)$, and let $\GG$ be a graph class containing $G$.
	We say that $N^t(e)$ \emph{has an extension $\ext^t_v(e)$ in $\GG$ (along $v$)} if there is a graph $H \in \GG$ such that $N^t(e) \cup \ext^t_v(e) \cong N^t_H(w)$, where $w$ is some node of $H$.
	For convenience, we will identify $v$ and $w$ (and other nodes and edges in isomorphic neighborhoods), allowing us, e.g., to specify the output label of $(v,e)$ in graph $H$.
	Analogously, we say that $N^{t-1}(v)$ \emph{has an extension $\ext^t_e(v)$ in $\GG$ (along $e$)} if there is a graph $H \in \GG$ such that $N^{t-1}(v) \cup \ext^t_e(v) \cong N^t_H(e')$, where $e'$ is some edge of $H$, and we will identify nodes (and edges) similarly as above.

	\paragraph{$t$-Independence}
	As mentioned in Section\ \ref{sec:lift}, our speedup results require a property called $t$-independence.
	This property essentially only exists in graph classes with sufficiently high girth which is the reason (together with the difficulty of determining neighborhood topologies if the girth is too small) why we restrict attention to these classes of graphs.
	Roughly speaking, $t$-independence is satisfied if for any fixed radius-$(t-1)$ neighborhood of a node $v$, the set of extensions along one incident edge $e$ is independent of the sets of extensions along the other incident edges (i.e., fixing an extension along one edge does not influence which extensions are possible along the other edges), and if a similar statment holds for edge neighborhoods.
	See Figure \ref{fig:indep} for an illustration.
	\begin{figure}
		\centering
		\includegraphics[scale=0.8]{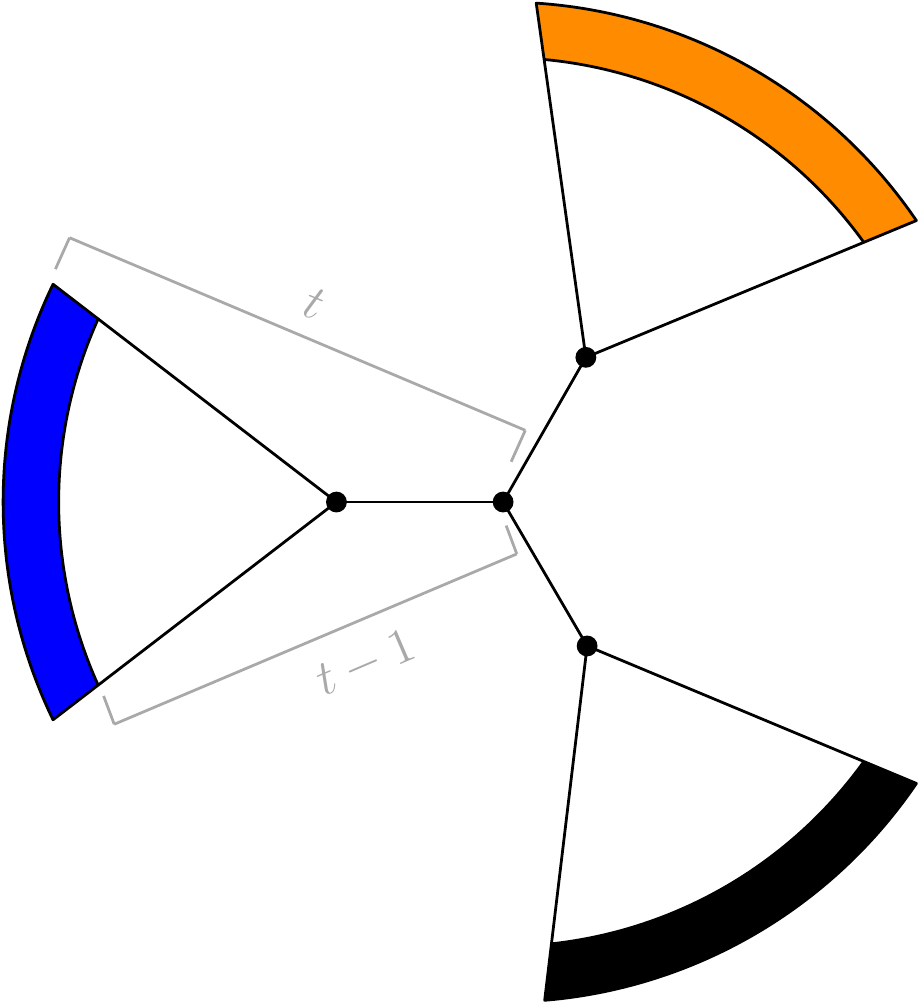}
		\caption{One of the two requirements for $t$-independence is that, for any node $v$, the sets of extensions of $N^{t-1}(v)$ along the different incident edges are independent of each other. In other words, for the depicted node of degree $3$, if we fix one of the possibilities for the topology and the input labels in the blue area, then this does not change which topologies and input label combinations are possible in the orange and the black area, and vice versa.}
		\label{fig:indep}
	\end{figure}
	Formally, we define $t$-independence as follows.

	Let $\GG_{n, \Delta}$ be a $\Sigma$-input-labeled graph class consisting of graphs with $n$ nodes and maximum degree $\Delta$ and let $t$ be a positive integer.
	We say that $\GG_{n, \Delta}$ is \emph{$t$-independent} if for any graph $G \in \GG_{n, \Delta}$, any node $v \in V(G)$, and any edge $e = \{ u, v \} \in E(G)$, the following two properties are satisfied:
	\begin{enumerate}
		\item Set $\extra^t_v(e) := \{ \ext^t_{w,H}(e') \mid H \in \GG_{n, \Delta}, e'\in E(H), \textrm{ and } H, e' \textrm{ satisfy } N^t_H(e') \cong N^t_G(e) \textrm{ where} \\ \textrm{$v$ and $w$ } \textrm{are corresponding nodes under this isomorphism} \}$, and define $\extra^t_u(e)$ analogously. Then, for each element $X \in \extra^t_v(e)$ and each element $X' \in \extra^t_u(e)$, there is a graph $H \in \GG_{n, \Delta}$ and an edge $e'\in E(H)$ such that $N^t_H(e') \cong N^t_G(e)$, $\ext^t_{w,H}(e') \cong X$, and $\ext^t_{x,H}(e') \cong X'$, where, under the given isomorphism, $w$ corresponds to $v$, and $x$ to $u$.
		\item For each edge $e' \in E(G)$ incident to $v$, set $\extra^t_{e'}(v) := \{ \ext^t_{e'',H}(w) \mid H \in \GG_{n, \Delta}, w \in V(H), \textrm{ and } H, w \textrm{ satisfy } N^{t-1}_H(w) \cong N^{t-1}_G(v) \textrm{ where $e'$,$e''$ are corresponding edges under this}$ $\textrm{isomorphism} \}$. Then, for each indexed family $(X_{e'})_{e' \in E(G):v \in e'}$ with $X_{e'} \in \extra^t_{e'}(v)$, there is a graph $H \in \GG_{n, \Delta}$ and a node $w \in V(H)$ such that $N^{t-1}_H(w) \cong N^{t-1}_G(v)$ and, for all $e'$ incident to v, $\ext^t_{e'',H}(w) \cong X_{e'}$, where, under the given isomorphism, $e''$ corresponds to $e'$.
	\end{enumerate}
	While this exact definition of $t$-independence is cumbersome, the intuition behind it makes it straightforward to check that $t$-independence is satisfied for the usual symmetry breaking inputs that do not include unique IDs, such as node colorings, edge colorings\footnote{We assume all symmetry breaking inputs to be given in the natural way, i.e., input for a node $v$ to be encoded at all $(v,e')\in B(G)$, and input for an edge $e$ to be encoded at all $(u,e)\in B(G)$. In particular, in a $0$-round algorithm each node is aware of the colors and orientations of all incident edges (otherwise $t$-independence might not be satisfied).}, edge orientations, or combinations thereof (in common graph classes with graphs of girth at least $2t+2$).

	If $\GG_{n, \Delta}$ is $t'$-independent for each $1 \leq t' \leq t$, then we say that $\GG_{n, \Delta}$ is \emph{$(\leq t)$-independent}.
	If $\GG$ is a graph class such that, for each non-negative $n, \Delta$, $\GG_{n, \Delta}$ is $t$-independent (resp.\ $(\leq t)$-independent), then we say that $\GG$ is $t$-independent (resp.\ $(\leq t)$-independent).
	Note that if $\GG_{n, \Delta}$ is empty, it is trivially $t$-independent for any positive $t$.

	\section{The Speedup Theorem}\label{sec:speedup}
	In this section, we will present our speedup results and apply them to several problems.
	We start by defining our automatic speedup that produces a sequence of problems with decreasing runtimes, and subsequently prove our main theorem (Section\ \ref{sec:thethe}).
	In Section\ \ref{sec:maxisim}, we introduce an important simplification technique that reduces the complexity of the descriptions of the problems in our sequence (by transforming the problems) and show that this technique is compatible with the main theorem.
	Then we prove that the main theorem can be extended to settings with unique IDs if we restrict attention to order-invariant algorithms (Section\ \ref{sec:oia}).
	In Sections\ \ref{sec:sinkless} and \ref{sec:colred}, we will see how to obtain two known results by applying our speedup, confirming the viability and generality of the speedup technique.
	Finally, in Section\ \ref{sec:app2col}, we will examine the effect of our speedup on weak $2$-coloring, giving some intuition for the generalization of weak $2$-coloring that is essential for our lower bound proof in Section\ \ref{sec:weak}. 

	\subsection{The Theorem}\label{sec:thethe}
	In the following, we describe how to transform a given problem $\Pi$ into a problem $\Pi_1$ that can be solved one round faster.
	Our speedup consists of two steps: First we transform $\Pi$ into a problem $\Pi_{1/2}$ that, in some sense, can be solved half a round faster, then we transform $\Pi_{1/2}$ into $\Pi_1$.
	Recall that a $t$-round distributed algorithm (for $\Pi$) is nothing else than a function that maps each possible radius-$t$ neighborhood $N^t(v)$ to a tuple of outputs for $v$.
	By saying that $\Pi_{1/2}$ can be solved half a round faster, we mean that there is a distributed algorithm for $\Pi_{1/2}$ where each node $v$ looks only at a neighborhood that is smaller than its radius-$t$ neighborhood $N^t(v)$, but larger than its radius-$(t-1)$ neighborhood $N^{t-1}(v)$.
	More precisely for deciding on the output at $(v,e)$, node $v$ looks only at the radius-$t$ neighborhood $N^t(e)$ of $e$.
	An alternative way to look at this is to consider the \emph{edges} of the input graph as the computing entities in an algorithm for $\Pi_{1/2}$; each edge $e = \{u, v\}$ then decides on the outputs at $(u,e)$ and $(v,e)$.
	This highlights the inherent \emph{duality} of the two steps in our speedup: both steps are essentially the same with the difference that the role of nodes and edges is swapped. 

	\paragraph{Deriving Problems}
	Let $\Pi$ be a problem, and let $\OO$ be the set of output labels and $f, g, h$ the functions used to define $\Pi$.
	We define the problem $\Pi_{1/2}$ by specifying the set $\OO_{1/2}$ of output labels and the three required functions $f_{1/2}$, $g_{1/2}$, and $h_{1/2}$ as follows.
	
	We set $\OO_{1/2} := 2^{\OO}$ and $f_{1/2}(\Delta) := 2^{f(\Delta)}$, and we define $g_{1/2}(\Delta)$ as the set of all multisets $\{ Y, Z \}$, where $Y,Z \in f_{1/2}(\Delta)$, with the following property:
	\begin{enumerate}
		\item\label{prop:match} For any $y \in Y$, $z \in Z$, the multiset $\{y,z\}$ is contained in $g(\Delta)$.
	\end{enumerate}
	We define $h_{1/2}(\Delta)$ as the set of all multisets $\{ Y_1, \dots, Y_i \}$, where $i \leq \Delta$ and $Y_1, \dots, Y_i \in f_{1/2}(\Delta)$, with the following property:
	\begin{enumerate}[resume]
		\item\label{prop:eych} There exist elements $y_1 \in Y_1, \dots, y_i \in Y_i$ such that the multiset $\{y_1, \dots, y_i\}$ is contained in $h(\Delta)$.
	\end{enumerate}

	\noindent Since $f(\Delta)$ is finite for all $\Delta \in \NN$, also $f_{1/2}(\Delta)$ is finite.
	Hence, $\Pi_{1/2}$ indeed satisfies the definition of a problem.
	Similarly to how we derived $\Pi_{1/2}$ from $\Pi$, we will now derive $\Pi_1$ from $\Pi_{1/2}$.
	More precisely, $\Pi_1$ is defined as follows.

	We set $\OO_1 := 2^{\OO_{1/2}}$ and $f_1(\Delta) := 2^{f_{1/2}(\Delta)}$, and we define $g_1(\Delta)$ as the set of all multisets $\{ Y, Z \}$, where $Y,Z \in f_1(\Delta)$, with the following property:
	\begin{enumerate}[resume]
		\item\label{prop:gee} There exist elements $y \in Y, z \in Z$ such that the multiset $\{ y, z\}$ is contained in $g_{1/2}(\Delta)$.
	\end{enumerate}
	We define $h_1(\Delta)$ as the set of all multisets $\{ Y_1, \dots, Y_i \}$, where $i \leq \Delta$ and $Y_1, \dots, Y_i \in f_1(\Delta)$, with the following property:
	\begin{enumerate}[resume]
		\item\label{prop:matchtwo} For any $y_1 \in Y_1, \dots, y_i \in Y_i$, the multiset $\{y_1, \dots, y_i\}$ is contained in $h_{1/2}(\Delta)$.
	\end{enumerate}

	\noindent With the same reasoning as for $\Pi_{1/2}$, we see that $\Pi_1$ satisfies the definition of a problem.
	For an illustration of the definitions, we refer to the concrete examples in Sections\ \ref{sec:sinkless}--\ref{sec:app2col}.
	In the same way as we derived $\Pi_{1/2}$ and $\Pi_1$ from $\Pi$, we can derive problems $\Pi_{3/2}$ and $\Pi_2$ from $\Pi_1$.
	In general, set $\Pi_0 := \Pi$, and for any positive integer $k$, define recursively $\Pi_{k+1/2} := (\Pi_k)_{1/2}$ and $\Pi_{k+1} := (\Pi_k)_1$.
	
	The intuition behind the definition of $\Pi_{1/2}$ (and the proof of Theorem\ \ref{thm:speedup}) is that the restrictions for $g_{1/2}(\Delta)$ and $h_{1/2}(\Delta)$ given in Properties\ \ref{prop:match} and\ \ref{prop:eych} are \emph{just} weak enough that an algorithm that looks only at radius-$(t-1)$ edge neighborhoods can infer correct outputs by simulating an algorithm for $\Pi$ on the radius-$t$ neighborhoods obtained by extending the radius-$(t-1)$ edge neighborhoods (with all possible extensions).
	The tightness of the two properties makes it possible that, conversely, we can also infer a (half-round slower) algorithm for $\Pi$ from an algorithm for $\Pi_{1/2}$.
	A dual version of the above arguments holds for the relation between $\Pi_{1/2}$ and $\Pi_1$.
	In the following, we show how to derive the algorithms for $\Pi_{1/2}$ and $\Pi_1$ mentioned above from an algorithm for $\Pi$ in a black box manner (the converse direction will be part of the proof of Theorem\ \ref{thm:speedup}).

	\paragraph{Algorithm Speedup}
	Let $\AAA$ be an arbitrary algorithm that solves some problem $\Pi$ on some graph class $\GG$.
	Let $t_{n, \Delta}$ be the worst-case runtime of $\AAA$ on the graph class $\GG_{n, \Delta}$. 
	We define an algorithm $\AAA_{1/2}$ for $(\Pi_{1/2}, \GG)$ as follows.

	Consider an arbitrary graph $G \in \GG$, and let $n$ denote the number of nodes of $G$, and $\Delta$ the maximum degree.
	In the following, we describe how an arbitrary node $v \in V(G)$ executing $\AAA_{1/2}$ decides on the output label it will assign to $(v,e) \in B(G)$, where $e = \{ u, v \}$ is some arbitrary edge incident to $v$.
	Node $v$ starts executing Algorithm $\AAA_{1/2}$ by collecting the radius-$t_{n, \Delta}$ neighborhood of edge $e$.
	Note that $\AAA_{1/2}$ has access to $n$ and $\Delta$ and therefore can determine $t_{n, \Delta}$.
	Then, for each output label $o \in f(\Delta)$, node $v$ determines whether $N^{t_{n, \Delta}}(e)$ has an extension $\ext^t_v(e)$ in $\GG_{n, \Delta}$ such that $v$ would assign output $o$ to $(v,e)$ according to $\AAA$.
	Let $O^*$ be the set of all $o$ for which such an extension exists.
	In other words, $v$ simulates $\AAA$ on each radius-$t$ node neighborhood occurring in $\GG_{n, \Delta}$ that can be obtained by extending $N^{t_{n, \Delta}}(e)$ in direction of $v$ (seen from $e$), and collects in $O^*$ all output labels that $\AAA$ outputs in each such neighborhood at the node-edge pair corresponding to $(v,e)$.
	Finally, node $v$ outputs $O^* \in 2^{f(\Delta)}$ at $(v,e)$.
	This concludes the description of $\AAA_{1/2}$.
	We will argue about the correctness and the runtime of $\AAA_{1/2}$ (and $\AAA_1$) in the proof of Theorem\ \ref{thm:speedup}.

	Analogously, we define an algorithm $\AAA_1$ as follows, by deriving it from $\AAA_{1/2}$.
	Again, we focus on the output $v$ assigns to $(v, e) \in B(G)$ according to $\AAA_1$, where $G \in \GG_{n, \Delta}$.
	Similarly to before, $v$ starts executing Algorithm $\AAA_1$ by collecting the radius-$(t_{n, \Delta} - 1)$ neighborhood of $v$.
	Then, for each output label $o \in f_{1/2}(\Delta)$, node $v$ determines whether $N^{t_{n, \Delta} - 1}(v)$ has an extension $\ext^t_e(v)$ in $\GG_{n, \Delta}$ such that $v$ would assign output $o$ to $(v,e)$ according to $\AAA$.
	Finally, $v$ collects all $o$ for which such an extension exists in a set $O^* \in 2^{f_{1/2}(\Delta)}$, and outputs $O^*$ at $(v, e)$.
	Now we are set to prove our main speedup result.

	\begin{theorem}\label{thm:speedup}
		Consider some arbitrary problem $\Pi$ and the derived problem $\Pi_1$.
		Let $\GG_{n, \Delta}$ be a $\Sigma$-input-labeled graph class consisting of $n$-node graphs with maximum degree $\Delta$.
		Assume that $\GG_{n, \Delta}$ is $t$-independent for some positive integer $t$ and contains only graphs of girth at least $2t+2$.
		Then, the following two statements are equivalent:
		\begin{itemize}
			\item[(1)] There is an algorithm solving $(\Pi,\GG_{n, \Delta})$ in time $t$. 
			\item[(2)] There is an algorithm solving $(\Pi_1, \GG_{n, \Delta})$ in time $t-1$.
		\end{itemize}
	\end{theorem}

	\begin{proof}
		We first show that (1) implies (2).
		Let $\AAA$ be an algorithm solving $(\Pi,\GG_{n, \Delta})$ in time $t$.
		Consider the algorithms $\AAA_{1/2}$ and $\AAA_1$ derived from $\AAA$ in the manner described above.
		We argue that $\AAA_{1/2}$ solves $(\Pi_{1/2},\GG_{n, \Delta})$ and that for deciding on the output assigned to some node-edge pair $(v,e)$ in the execution of $\AAA_{1/2}$, node $v$ only looks at $N^t(e)$.
		Furthermore, we argue that $\AAA_1$ solves $(\Pi_1,\GG_{n, \Delta})$ in time $t-1$.
		The statement that $v$ only looks at $N^t(e)$ directly follows from the definition of $\AAA_{1/2}$; analogously, the runtime of $t-1$ follows from the description of $\AAA_1$.
		Hence, what is left to show is that the two derived algorithms actually solve the two derived realized problems.
		We start with $\AAA_{1/2}$ and $(\Pi_{1/2},\GG_{n, \Delta})$.
		
		Let $G$ be an arbitrary graph in $\GG_{n, \Delta}$, and let $v \in V(G)$ be an arbitrary node, $e_1, \dots, e_{d(v)} \in E(G)$ the edges incident to $v$, and $u \in V(G)$ the other endpoint of $e = e_1$.
		Furthermore, let $O_{v, e}$ denote the output $\AAA_{1/2}$ assigns to $(v,e)$, and let the output for other node-edge pairs be denoted analogously.
		We have to show that A) the multiset $\{O_{v,e}, O_{u,e}\}$ is contained in $g_{1/2}(\Delta)$, and B) the multiset $\{ O_{v, e_1}, \dots, O_{v, e_{d(v)}} \}$ is contained in $h_{1/2}(\Delta)$.
		
		For A), observe that, by the definition of $\AAA_{1/2}$, for any two output labels $o \in O_{v,e}$ and $o' \in O_{u,e}$, $N^t(e)$ has both an extension $\ext^t_v(e)$ in $\GG_{n, \Delta}$ such that $\AAA$ would output $o$ at $(v,e)$ and an extension $\ext^t_u(e)$ in $\GG_{n, \Delta}$ such that $\AAA$ would output $o'$ at $(u,e)$.
		Due to the $t$-independence of $\GG_{n, \Delta}$, there is a graph $H \in \GG_{n, \Delta}$ such that the radius-$(t+1)$ neighborhood of some node $e' = \{ w, x \} \in E(H)$ is isomorphic to $N^t(e) \cup \ext^t_v(e) \cup \ext^t_u(e)$, which implies that, in $H$, $\AAA$ would output $o$ at $(w, e')$ \emph{and} $o'$ at $(x, e')$.
		Since $\AAA$ solves $(\Pi,\GG_{n, \Delta})$, it follows that the multiset $\{o, o'\}$ is contained in $g(\Delta)$.
		We conclude that for any $o \in O_{v,e}, o' \in O_{u,e}$, the multiset $\{o, o'\}$ is contained in $g(\Delta)$, which shows that the multiset $\{ O_{v,e}, O_{u,e} \}$ is contained in $g_{1/2}(\Delta)$, by Property\ \ref{prop:match} in the definition of $\Pi_{1/2}$.

		For B), consider $N^t_G(v)$, and let $o_1, \dots, o_{d(v)}$ be the outputs $v$ would assign to $(v, e_1), \dots, (v, e_{d(v)})$, respectively, when executing $\AAA$.
		Observe that, for any $1 \leq i \leq d(v)$, $N^t_G(e_i)$ has an extension in $\GG_{n, \Delta}$ such that $\AAA$ would output $o_i$ at $(v, e_i)$, namely $N^t_G(v) \setminus N^t_G(e_i)$.
		Hence, by the definition of $\AAA_{1/2}$, we have $o_i \in O_{v, e_i}$.
		Since the multiset $\{o_1, \dots, o_{d(v)}\}$ is contained in $h(\Delta)$ (due to $\AAA$ solving $(\Pi,\GG_{n, \Delta})$), it follows that the multiset $\{ O_{v, e_1}, \dots, O_{v, e_{d(v)}} \}$ is contained in $h_{1/2}(\Delta)$, by Property\ \ref{prop:eych} in the definition of $\Pi_{1/2}$.
		This concludes the proof showing that $\AAA_{1/2}$ solves $(\Pi_{1/2},\GG_{n, \Delta})$.
		
		The proof for showing that $\AAA_1$ solves $(\Pi_1,\GG_{n, \Delta})$ is analogous.
		Note that this proof does not rely in any way on the fact that $\Pi_{1/2}$ and $\AAA_{1/2}$ are derived from another problem, resp.\ algorithm.
		The proof works just as well if we replace $\Pi_{1/2}$ by any other problem and $\AAA_{1/2}$ by any other algorithm that solves the problem on $\GG_{n, \Delta}$ and has the property that each node $v$ decides on the output for $(v,e)$ by only looking at $N^t(e)$.

		Now, we show that (2) implies (1).
		Let $\AAA^*$ be an algorithm solving $(\Pi_1,\GG_{n, \Delta})$ in time $t-1$.
		We derive an algorithm $\AAA^*_{-1/2}$ as follows.
		For each graph $G \in \GG_{n, \Delta}$ and each edge $e = \{ u, v \} \in E(G)$, let $O^G_{v, e}$ and $O^G_{u, e}$ be the outputs $\AAA^*$ assigns to $(v, e)$ and $(u, e)$, respectively, when executed (by $v$) on $G$. 
		Since $\AAA^*$ solves $(\Pi_1,\GG_{n, \Delta})$, there exist elements $o \in O^G_{v, e}, o' \in O^G_{u, e}$ such that the multiset $\{ o, o' \}$ is contained in $g_{1/2}(\Delta)$, by Property\ \ref{prop:gee} in the definition of $\Pi_1$.
		We define algorithm $\AAA^*_{-1/2}$ to output $o$ at $(v,e)$ and $o'$ at $(u,e)$, when executed on $G$; thereby $\AAA^*_{-1/2}$ trivially satisfies one of the two conditions for solving $(\Pi_{1/2},\GG_{n, \Delta})$ (the one concerning $g_{1/2}(\Delta)$).
		By Property\ \ref{prop:matchtwo} in the definition of $\Pi_1$ (and the fact that $\AAA^*$ solves $\Pi_1$ on $G$), algorithm $\AAA^*_{-1/2}$ also satisfies the other condition, concerning $h_{1/2}(\Delta)$.
		Hence $\AAA^*_{-1/2}$ solves $(\Pi_{1/2},\GG_{n, \Delta})$.
		Observe that due to the design of $\AAA^*_{-1/2}$, each node $v$ can determine the output at $(v,e)$ according to $\AAA^*_{-1/2}$ by looking only at $N^t(e)$.

		Now, using an analogous argumentation, we can define an algorithm $\AAA^*_{-1}$ that solves $(\Pi,\GG_{n, \Delta})$, where each node $v$ can determine the output at $(v,e)$ by looking only at $N^t(v)$.
		Hence, there is an algorithm solving $(\Pi,\GG_{n, \Delta})$ in time $t$.
	\end{proof}

	Theorem\ \ref{thm:speedup} provides an explicit version of our speedup guarantees for graph classes $\GG_{n, \Delta}$.
	We formulate the implications for general graph classes in the following corollary of Theorem\ \ref{thm:speedup}.

	\begin{corollary}\label{cor:allgg}
		Consider some arbitrary problem $\Pi$, and let $\GG$ be a $\Sigma$-input-labeled graph class and $T: \NN \times \NN \rightarrow \NN^+$ some function. 
		If, for all $n, \Delta \in \NN$, it holds that $\GG_{n, \Delta}$ is $T(n, \Delta)$-independent and contains only graphs of girth at least $2T(n, \Delta)+2$, then the following two statements are equivalent:
		\begin{itemize}
			\item[(1)] There is an algorithm solving $(\Pi,\GG)$ in time $T(n, \Delta)$. 
			\item[(2)] There is an algorithm solving $(\Pi_1, \GG)$ in time $T(n, \Delta)-1$.
		\end{itemize}
	\end{corollary}

	Note that Theorem \ref{thm:speedup} and Corollary \ref{cor:allgg} do not contradict the existence of known gaps \cite{chang16separation, DBLP:conf/focs/ChangP17, DBLP:conf/stoc/NaorS93} in the time complexity landscape of LCL problems since the problem obtained after applying our speedup a non-constant number of times is not necessarily an LCL problem anymore (but is still locally checkable).

	\subsection{Simplification via Maximality}\label{sec:maxisim}	
	In this section, we will present a technique that simplifies the derived problems $\Pi_{1/2}$ and $\Pi_1$ without affecting the correctness of Theorem\ \ref{thm:speedup}.
	Consider the (common) case that the input labeling of a considered graph class contains some symmetry breaking information for the two endpoints of each edge, i.e., any two nodes $u$, $v$ connected by an edge $e$ can select the same of the two endpoints.
	We will consider the simplest case that the desired symmetry breaking can be achieved by both nodes by just looking at their $0$-round input information, i.e., the respective input graph contains edge orientations (where the orientation of an edge $e = \{u, v\}$ is encoded in both input labels assigned to $(u,e)$ and $(v,e)$).
	Note that if, instead, the nodes have to look up to a distance of $1$ (or, more generally, some constant $c$), as, e.g., in the case of an input node coloring or input IDs, then the simplification technique we present still works for all cases except for the speedup of a $1$-round (resp.\ any $(\geq c)$-round) algorithm to a $0$-round (resp.\ one-round-faster) algorithm.
	
	The idea behind the simplification technique is that we can reduce the number of output labels the derived algorithm for $\Pi_{1/2}$ can use, by making the set $g_{1/2}(\Delta)$ smaller; as we will show this can be done in a way that ensures that the algorithm (or, more precisely, its adapted version) can still solve (the new version of) $\Pi_{1/2}$ in the same runtime as before.
	More specifically, observing that the definition of $\Pi_{1/2}$ contains an existential and a universal statement, we can replace each multiset $\{Y,Z\}$ in $g_{1/2}(\Delta)$ by some $\{Y',Z\}$ with $Y' \supsetneq Y$ if $\{Y', Z\}$ still satisfies the universal Property\ \ref{prop:match} in the description of $\Pi_{1/2}$; the correctness of the existential Property\ \ref{prop:eych} is not affected by this and $g_{1/2}(\Delta)$ shrinks.
	The symmetry breaking on edges mentioned above is a necessary ingredient for performing this process in a consistent manner (in the dual case for $\Pi_1$, we need analogously a symmetry breaking for nodes; however, this is already provided by the port numbers).
	As we will see in Sections\ \ref{sec:sinkless}--\ref{sec:app2col}, this approach can lead to a significant decrease in usable outputs.
	For a simplified problem, we may compress the problem description by assuming that the respective $\mathcal O$ and $f(\Delta)$ contain only those outputs that can actually occur in a correct solution (i.e., that occur in both some multiset contained in $g(\Delta)$ and some multiset contained in $h(\Delta)$).
	
	\paragraph{Simplified Problems}
	We obtain the simplified version $\Pi'_{1/2}$ of $\Pi_{1/2}$ by replacing Property\ \ref{prop:match} in the description of $\Pi_{1/2}$ by the following extension:
	\begin{enumerate}[resume]
		\item\label{prop:replacematch} For any $y \in Y$, $z \in Z$, the multiset $\{y,z\}$ is contained in $g(\Delta)$ and the sets $Y$ and $Z$ are maximal with this property, i.e., for any $Y' \supsetneq Y$, the multiset $\{Y', Z\}$ does not satisfy Property\ \ref{prop:match}, and for any $Z' \supsetneq Z$, the multiset $\{Y, Z'\}$ does not satisfy Property\ \ref{prop:match}.
	\end{enumerate}

	\noindent Analogously, we obtain the simplified version of $\Pi_1$ by replacing Property\ \ref{prop:matchtwo} by the following extension:
	\begin{enumerate}[resume]
		\item For any $y_1 \in Y_1, \dots, y_i \in Y_i$, the multiset $\{y_1, \dots, y_i\}$ is contained in $h_{1/2}(\Delta)$ and the sets $Y_1, \dots, Y_i$ are maximal with this property, i.e., for any $1 \leq j \leq i$, it holds that, for any $Y'_j \supsetneq Y_j$, the multiset $\{Y_1, \dots, Y_{j-1}, Y'_j, Y_{j+1}, \dots, Y_i\}$ does not satisfy Property\ \ref{prop:matchtwo}.
	\end{enumerate}
	However, just performing the latter replacement as it is will result in a (simplified) problem that is derived from $\Pi_{1/2}$, not from $\Pi'_{1/2}$.
	Hence, we define $\Pi'_1$ by not only performing the above replacement, but also replacing $g_{1/2}$ by $g'_{1/2}$ in Property\ \ref{prop:gee}, where $g'_{1/2}$ is the analogue to $g_{1/2}$ in $\Pi'_{1/2}$ (and actually the only part in which the definitions of $\Pi_{1/2}$ and $\Pi'_{1/2}$ differ).
	Now, similarly to above we can define simplified problems recursively for any positive $k$ by setting $\Pi'_{k+1/2} := (\Pi'_k)'_{1/2}$ and $\Pi'_{k+1} := (\Pi'_k)'_1$.
	Note that for simplicity, we write, e.g., $g_{1/2}$ instead of $g'_{1/2}$ when talking about the simplified problem $\Pi'_{1/2}$ if it is clear which problem is considered.

	In the following theorem, we show that Theorem\ \ref{thm:speedup} still holds if we replace a derived problem by its simplified version.
	Naturally, this implies that also an analogous version of Corollary\ \ref{cor:allgg} holds for our simplified problems.
	Moreover, from the proof of Theorem\ \ref{thm:primespeedup}, it follows that we can arbitrarily decide after which performed problem derivation we apply the simplification technique and after which we do not, without affecting the runtime of the derived problem.

	\begin{theorem}\label{thm:primespeedup}
		If we augment the conditions of Theorem\ \ref{thm:speedup} by the requirement that the inputs for the graphs from $\GG_{n, \Delta}$ contain an orientation for each edge, then the following two statements are equivalent:
		\begin{itemize}
			\item[(1)] There is an algorithm solving $(\Pi,\GG_{n, \Delta})$ in time $t$. 
			\item[(2)] There is an algorithm solving $(\Pi'_1, \GG_{n, \Delta})$ in time $t-1$.
		\end{itemize}
	\end{theorem}
	
	\begin{proof}
		Recall that, as observed in the proof of Theorem\ \ref{thm:speedup}, the speedup of one round is actually obtained by two \emph{independent} speedups, going from $\Pi$ to $\Pi_{1/2}$ and from $\Pi_{1/2}$ to $\Pi_1$.
		In particular, the fact that $\Pi_{1/2}$ is a derived problem is not relevant for the second speedup.
		Therefore (and due to the two speedups being analogous), it is enough to argue that if there is an algorithm $\AAA$ solving $(\Pi_{1/2},\GG_{n, \Delta})$ by looking only at radius-$t$ edge neighborhoods, then there is also an algorithm $\AAA'$ solving $(\Pi'_{1/2},\GG_{n, \Delta})$ by looking only at radius-$t$ edge neighborhoods, and vice versa.
		Going from $\AAA'$ to $\AAA$ is easy: since $\Pi'_{1/2}$ is a more restrictive version of $\Pi_{1/2}$, any algorithm solving $(\Pi'_{1/2},\GG_{n, \Delta})$ also solves $(\Pi_{1/2},\GG_{n, \Delta})$, hence setting $\AAA := \AAA'$ is sufficient.
	
		For the other direction, given an algorithm $\AAA$ as described above, we transform it into an algorithm $\AAA'$ as described above, as follows.
		Consider an arbitrary edge $e = \{ u, v \}$ in some graph $G \in \GG_{n, \Delta}$.
		From the orientation of $e$, nodes $u$ and $v$ infer an ordering of the set $\{u, v\}$ in some fixed way (so, both $u$ and $v$ infer the same ordering, in $0$ rounds).
		Then, $u$ and $v$ compute the outputs $O_{u,e}$ and $O_{v,e}$ of $\AAA$ at $(u,e)$ and $(v,e)$, respectively.
		Since these two outputs are uniquely defined by $N^t(e)$ (and $\AAA$), both $u$ and $v$ can compute \emph{both} outputs by only looking at $N^t(e)$.
		Now both nodes do the following (internally, in $0$ additional rounds), where we assume w.l.o.g.\ that $u$ is first in the computed ordering.
		First, in a fixed way\footnote{More formally, this is to say that $u$ and $v$ apply a deterministic algorithm to infer which maximal superset of $O_{u,e}$ is picked (based on knowing $O_{u,e}$ and $O_{v,e}$), so that both nodes pick \emph{the same} superset.}, they pick a set $O \supseteq O_{u,e}$, $O \in f_{1/2}(\Delta)$ that is maximal among all such sets with the property that the multiset $\{ O, O_{v,e} \}$ satisfies Property\ \ref{prop:match} in the definition of $\Pi_{1/2}$.
		Then, analogously, they pick a set $O' \supseteq O_{v,e}$, $O' \in f_{1/2}(\Delta)$ that is maximal among all such sets with the property that the multiset $\{ O, O' \}$ satisfies Property\ \ref{prop:match}.
		Finally, $u$ outputs $O$ at $(u,e)$, and $v$ outputs $O'$ at $(v,e)$.
		Now it is straightforward to check that algorithm $\AAA'$, defined in this way, actually solves $(\Pi'_{1/2},\GG_{n, \Delta})$, and each node $v$ can determine the output at $(v,e)$ by only looking at $N^t(e)$.
	\end{proof}

	\subsection{Order-Invariant Algorithms}\label{sec:oia}
	As mentioned in Section\ \ref{sec:lift}, if we are interested in runtime bounds as a function of $\Delta$, then restricting our attention to order-invariant algorithms allows our speedup to be applied to models where nodes are equipped with unique IDs, such as the $\LOCAL$ model.
	The following theorem extends Theorems\ \ref{thm:speedup} and\ \ref{thm:primespeedup} (and, consequently, the respective corollaries) to the case of order-invariant algorithms and is an important building block in our lower bound proof for weak $2$-coloring.
	Recall that an order-invariant algorithm \cite{DBLP:conf/stoc/NaorS93} is defined as an algorithm where each node's output only depends on the relative order of the IDs it sees (and possibly other input information), but not their absolute values (i.e., if the input ID assignment in the view of a node $v$ is changed to another assignment with the same ordering of the nodes according to the IDs, then the output of $v$ must remain the same).

	\begin{theorem}\label{thm:ordinv}
		Consider some arbitrary problem $\Pi$ and the derived problems $\Pi_1$ and $\Pi'_1$.
		Let $\GG_{n, \Delta}$ be a $\Sigma$-input-labeled graph class consisting of $n$-node graphs with maximum degree $\Delta$.
		Assume that $\GG_{n, \Delta}$ is $t$-independent for some positive integer $t$ and contains only graphs of girth at least $2t+2$.
		Let $S \subseteq \NN^+$ be a finite set of identifiers satisfying $|S| \geq 4\Delta^{2t}$.
		Let $\GG'_{n, \Delta}$ be the class consisting of all\footnote{If $|S| < n$, then $\GG'_{n, \Delta}$ is empty. If one wants to avoid empty graph classes, one can simply require $|S| \geq \max (4\Delta^{2t}, n)$, here and in Lemma \ref{lem:supersuper}.} (input-labeled) graphs obtained by taking a graph from $\GG_{n, \Delta}$ and assigning unique identifiers from $S$ to the nodes of the graph.
		Then, the following statements are equivalent:
		\begin{itemize}
			\item[(1)] There is an order-invariant algorithm solving $(\Pi,\GG'_{n, \Delta})$ in time $t$. 
			\item[(2)] There is an order-invariant algorithm solving $(\Pi_1, \GG'_{n, \Delta})$ in time $t-1$.
			\item[(3)] There is an order-invariant algorithm solving $(\Pi'_1, \GG'_{n, \Delta})$ in time $t-1$.
		\end{itemize}
	\end{theorem}

	\begin{proof}
		We first show that (1) implies (2).
		Assume there is some order-invariant algorithm solving $(\Pi,\GG'_{n, \Delta})$ in time $t$, and for convenience, denote this algorithm by $\AAA$.
		Our theorem would follow from Theorem\ \ref{thm:speedup} if $\GG'_{n, \Delta}$ was $t$-independent; unfortunately, we only know that $\GG_{n, \Delta}$ is $t$-independent.
		We circumvent this issue by adapting our definition of the two algorithms $\AAA_{1/2}$ and $\AAA_1$.
		Recall the definition of $\AAA_{1/2}$.
		We change it only very slightly, as follows, again focusing on the output a node $v \in V(G)$ will assign to $(v, e) \in B(G)$, where $G \in \GG'_{n, \Delta}$.
		Directly after collecting the radius-$t$ neighborhood of edge $e$, node $v$ internally changes the IDs assigned to the nodes in $N^t(e)$.
		More precisely, in a fixed way, depending only on the \emph{relative} ID values (i.e., in an order-invariant way), each ID is changed (to some ID from $S$) so that the order of the nodes in $N^t(e)$ according to their IDs stays the same and any two of these IDs differ by at least $2\Delta^{t} + 1$.
		By our constraint on the size of $S$, this is possible as the number of nodes in $N^t(e)$ is at most $2\Delta^{t-1}$.
		After that, node $v$ proceeds as given in the definition of $\AAA_{1/2}$.
		Clearly, this new version of $\AAA_{1/2}$ is order-invariant and each node only looks at $N^t(e)$; in the following we argue that it also solves $(\Pi_{1/2},\GG'_{n, \Delta})$.

		For seeing that the argumentation in proof part B) of Theorem\ \ref{thm:speedup} still works, it is sufficient to observe that the set $O_{v,e}$ our modified $\AAA_{1/2}$ outputs at $(v,e)$ is a superset of the set the original $\AAA_{1/2}$ would output at $(v,e)$.
		This is the case due to the order-invariance of $\AAA$ and the fact that $2\Delta^{t}$ is at least as large as the number of nodes in $N^t(v) \setminus N^t(e)$ (so, for each actual ID assignment to $N^t(v)$, there is an ID assignment with the same order on the nodes in $N^t(v)$ that respects $v$'s internal ID choices for the nodes in $N^t(e)$).
		Looking at the individual steps in proof part A) of Theorem\ \ref{thm:speedup}, we see that, apart from changing $N^t(e)$ by updating the IDs as done by $v$ internally, the only argument we have to change due to our new definition of $\AAA_{1/2}$ is the argument where the $t$-independence is used.
		However, instead of using the (non-existent) $t$-independence of $\GG'_{n, \Delta}$, we can simply use the $t$-independence of $\GG_{n, \Delta}$, combined with the order-invariance of $\AAA$ and the fact that $2\Delta^{t}$ is at least as large as the number of nodes in $N^t(v) \setminus N^t(e)$ and $N^t(u) \setminus N^t(e)$ together (where $u$ is the other endpoint of $e$). An analogous adaptation of $\AAA_1$ (derived from the adapted $\AAA_{1/2}$) combined with an analogous proof adaptation yields our first implication.
		
		For the implication from (2) to (1), we can use an analogous argumentation to the one given in the proof of Theorem\ \ref{thm:speedup}.
		The only place where we have to be careful is when we choose elements $o \in O^g_{v,e}, o' \in O^G_{u,e}$ (such that $\{o,o'\} \in g_{1/2}(\Delta)$) during the definition of $\AAA^*_{-1/2}$ (and in the analogous place during the definition of $\AAA^*_{-1}$).
		Here, in order to ensure the order-invariance of $\AAA^*_{-1/2}$, we have to choose $o$ and $o'$ in a way that ensures that the same $o, o'$ are chosen for the same \emph{relative} ID assignments in $N^t(e)$.
		This is possible due to the order-invariance of $\AAA^*$.
		
		For the implication from (2) to (3), observe that the same algorithm transformation (``from $\AAA$ to $\AAA'$") as performed in the proof of Theorem\ \ref{thm:primespeedup} (where, in the current case, we infer the required edge orientations from the (relative) IDs) ensures that the obtained algorithm solves $(\Pi'_{1/2}, \GG'_{n, \Delta})$ in time $t-1$; moreover, as the outputs given by the transformed algorithm depend only on the outputs given by the initial algorithm and the inferred edge orientations, the order-invariance of the transformed algorithm follows from the order-invariance of the initial algorithm.
		Since the above statements hold analogously for the case of $(\Pi'_1, \GG'_{n, \Delta})$ (instead of $(\Pi'_{1/2}, \GG'_{n, \Delta})$), the desired implication follows.
		The implication from (3) to (2) is shown as in the proof of Theorem\ \ref{thm:primespeedup}.
	\end{proof}

	\subsection{Warm-up I: Sinkless Coloring and Sinkless Orientation}\label{sec:sinkless}
	In this section, we will apply our speedup technique to the problem of sinkless orientation.
	As we will see, our technique gives a deterministic version of the randomized speedup for sinkless orientation presented in \cite{Brandtlll}, showing that our speedup in some sense encompasses known techniques for concrete problems as special cases.
	In particular, the duality between sinkless orientation and sinkless coloring which is essential for the lower bound proof in \cite{Brandtlll} is automatically produced by our speedup.

	The problem of sinkless coloring (as it is given in \cite{Brandtlll}) is defined as follows:
	Given a $\Delta$-regular edge $\Delta$-colored graph, each node has to output an orientation for each incident edge such that the two endpoints of an edge agree on the orientation of the edge and each node has at least one outgoing edge.
	
	While our speedup does not require the edge coloring, we will add edge orientations to the input in order to be able to apply our simplification technique.
	Note that adding edge orientations or some other symmetry breaking information as input is \emph{required} in the deterministic case since the random bits used to break symmetry in the randomized case are not available and the edge coloring alone is not sufficient for breaking symmetry.
	
	While it is possible to set our initial problem $\Pi$ to be sinkless orientation (in which case the derived problem $\Pi_{1/2}$ is also sinkless orientation, roughly speaking because the output for sinkless orientation describes an \emph{edge} property), it is more natural to choose $\Pi$ to be the dual problem called sinkless coloring (since for this problem the output describes a \emph{node} property).
	Sinkless coloring is defined as follows:
	Given a $\Delta$-regular edge $\Delta$-colored graph, each node has to output one of the edge colors such that at least one of the two endpoints of any edge outputs a different color than the color of the edge.

	Due to our choice of having one output per element of $B(G)$, not per node, there is a simple (and canonical) way to encode this problem in our setting, even without an input edge coloring.
	Define $\Pi$ as follows (where an output of $1$ at $(v,e)$ indicates that $v$ chooses the color of $e$).
	Set
	\begin{align*}
		f(\Delta) &:= \OO := \{ 0, 1 \} \enspace,\\
		g(\Delta) &:= \{ \{ 0, 0 \}, \{ 0, 1 \} \} \enspace,\\ 
		h(\Delta) &:= \{ \{ 0, \dots, 0, 1 \} \} \enspace,
	\end{align*}
	where the only element of $h(\Delta)$ is a multiset with $\Delta$ elements.
	Applying our speedup transformation, we obtain for the derived problem $\Pi_{1/2}$
	\begin{align*}
		f_{1/2}(\Delta) &= \OO_{1/2} = \{ \quad \{\}, \quad \{0\}, \quad \{1\}, \quad \{0,1\} \quad \} \enspace,\\
		g_{1/2}(\Delta) &= \{ \quad \{ \{0\}, \{0\} \}, \quad \{ \{0\}, \{1\} \}, \quad \{ \{0\}, \{0,1\} \} \quad \} \enspace,\\
		h_{1/2}(\Delta) &= \{ \{ Y_1, \dots, Y_{\Delta} \} \mid Y_i \in f_{1/2}(\Delta) \textrm{ for all } 1 \leq i \leq \Delta, 1 \in Y_1, 0 \in Y_j \textrm{ for all } 2 \leq j \leq \Delta \} \enspace.
	\end{align*}
	Applying our simplification technique, we see that the only element of $g_{1/2}(\Delta)$ that is maximal is $\{ \{0\} \{0,1\} \}$. Hence, we obtain for $\Pi'_{1/2}$
	\begin{align*}
		f_{1/2}(\Delta) &= \OO_{1/2} = \{ \quad \{0\}, \quad \{0,1\} \quad \} \enspace,\\
		g_{1/2}(\Delta) &= \{ \quad \{ \{0\}, \{0,1\} \} \quad \} \enspace,\\ 
		h_{1/2}(\Delta) &= \{ \{ Y_1, \dots, Y_{\Delta} \} \mid Y_i \in f_{1/2}(\Delta) \textrm{ for all } 1 \leq i \leq \Delta, Y_1 = \{0,1\} \} \enspace.
	\end{align*}
	Writing $0$ for $\{0\}$ and $1$ for $\{ 0, 1 \}$, we can describe $\Pi'_{1/2}$ equivalently by
	\begin{align*}
		f_{1/2}(\Delta) &= \OO_{1/2} := \{ 0, 1 \} \enspace,\\
		g_{1/2}(\Delta) &= \{ \{0,1\} \} \enspace,\\ 
		h_{1/2}(\Delta) &= \{ \{ Y_1, \dots, Y_{\Delta} \} \mid Y_i \in \{0,1\} \textrm{ for all } 1 \leq i \leq \Delta, Y_1 = 1 \} \enspace.
	\end{align*}
	This specification describes exactly the problem of sinkless orientation, where an output of $1$ at $(v,e)$ indicates that $v$ orients $e$ away from itself, and an output of $0$ at $(v,e)$ that $v$ orients $e$ towards itself.
	Now, deriving $\Pi_{1/2}$ leads to the problem $\Pi_1$ specified by
	\begin{align*}
		f_1(\Delta) &= \OO_1 = \{ \quad \{\}, \quad \{0\}, \quad \{1\}, \quad \{0,1\} \quad \} \enspace,\\
		g_1(\Delta) &= \{ \{ Y, Z \} \mid Y, Z \in f_1(\Delta), 0 \in Y, 1 \in Z \} \enspace,\\ 
		h_1(\Delta) &= \{ \{ Y_1, \dots, Y_{\Delta} \} \mid Y_i \in f_1(\Delta)\setminus \{ \emptyset \} \textrm{ for all } 1 \leq i \leq \Delta, Y_1 = \{1\} \} \enspace.
	\end{align*}
	Observing that the only maximal element of $h_1(\Delta)$ is $\{ \{1\}, \{0,1\}, \dots, \{0,1\} \}$, we see that $\Pi'_1$ is given by
	\begin{align*}
		f_1(\Delta) &= \OO_1 = \{ \quad \{1\}, \quad \{0,1\} \quad \} \enspace,\\
		g_1(\Delta) &= \{ \quad \{ \{0,1\}, \{1\} \}, \quad \{ \{0,1\}, \{0,1\} \} \quad \} \enspace,\\ 
		h_1(\Delta) &= \{ \quad \{ \{1\}, \{0,1\}, \dots, \{0,1\} \} \quad \} \enspace.
	\end{align*}
	Writing $0$ for $\{0,1\}$ and $1$ for $\{1\}$, we see that we are back at sinkless coloring, i.e., $\Pi_1 = \Pi$.

	Applying Theorem\ \ref{thm:primespeedup} iteratively, we see that the existence of a $t$-independent graph class $G_{n, \Delta}$ of girth at least $2t+2$ implies that any $t$-round algorithm for sinkless coloring can be transformed into a $0$-round algorithm for sinkless coloring.
	Since the latter does not exist, whereas, for any $\Delta \geq 3$, such high-girth $t$-independent graph classes exist for some $t \in \Omega(\log n)$ (as can be shown using \cite[Chapter III, Theorem $1.4'$]{bollobas78extremal}, or, in the case of an input edge-coloring, \cite[Lemma 9]{Brandtlll}), we obtain an $\Omega(\log n)$ lower bound for sinkless coloring, and therefore also for sinkless orientation, for each $\Delta \geq 3$. 

	\subsection{Warm-up II: Color Reduction}\label{sec:colred}
	In this section, we will give an example of how to use our speedup to obtain a (well-known) upper bound.
	The problem we are interested in is the problem of properly $3$-coloring a ring in the $\LOCAL$ model.
	Cole and Vishkin \cite{DBLP:conf/stoc/ColeV86} and Goldberg et al.\ \cite{DBLP:conf/stoc/GoldbergPS87} gave an upper bound of $O(\logstar n)$ for this problem, which was subsequently proved to be tight by Linial \cite{DBLP:journals/siamcomp/Linial92}.
	The upper bound is based on the idea of interpreting the given unique IDs (which are bitstrings of length $O(\log n)$) as colors and then reduce the number of colors in each round of the distributed computation exponentially.
	Other methods of achieving such an exponential (or doubly exponential) color reduction were devised later (see, e.g., \cite{DBLP:conf/focs/Linial87, DBLP:conf/istcs/MayerNS95, DBLP:conf/stoc/SzegedyV93}).
	
	We show that our speedup provides a doubly exponential color reduction (on rings).
	In order to achieve this result, we make use of the technique for upper bounds outlined in Section\ \ref{sec:apply}: after deriving problem $\Pi_1$ from $\Pi$, we change $\Pi_1$ to a problem that is provably at least as hard as $\Pi_1$ and has a much simpler description.
	As we will see in the following, if $\Pi$ is the problem of $k$-coloring, then the resulting problem is $k'$-coloring for some $k'$ satisfying $k \in O(\log\log k')$.
	We will focus on the case of $k = 4$ before arguing about general $k$. 

	Formally, we can describe $4$-coloring as follows\footnote{Note that we use in the description that the input graphs we are interested in are rings. Since our notion of a problem, however, is independent of the actually considered input graphs, we still have to define the three functions $f, g, h$ for any $\Delta$. The problem descibed by our specification is also defined on graph classes that do not only contain rings, but it might be the case that there is no algorithm that solves the respective realized problem.}.
	\begin{align*}
		f(\Delta) &:= \OO := \{ 1, 2, 3, 4 \} \enspace,\\
		g(\Delta) &:= \{ \{ y, z \} \mid y, z \in \{1,2,3,4\}, y \neq z \} \enspace,\\ 
		h(\Delta) &:= \{ \{ y, y \} \mid 1 \leq y \leq 4 \} \enspace.
	\end{align*}
	For the derived and simplified problem $\Pi'_{1/2}$, we obtain
	\begin{align*}
		f_{1/2}(\Delta) &= \OO_{1/2} = \{ Y \subseteq \{ 1, 2, 3, 4 \} \mid 1 \leq |Y| \leq 3 \} \enspace,\\
		g_{1/2}(\Delta) &= \{ \{ Y, Z \} \mid Y, Z \in f_{1/2}(\Delta), Y \cap Z = \emptyset, Y \cup Z = \{1,2,3,4\} \} \enspace,\\ 
		h_{1/2}(\Delta) &= \{ \{ Y, Z \} \mid Y, Z \in f_{1/2}(\Delta), Y \cap Z \neq \emptyset \} \enspace.
	\end{align*}
	After the following derivation, we do not apply our simplification technique, hence obtaining problem $\Pi_1$ characterized by
	\begin{align*}
		f_1(\Delta) &= \OO_1 = 2^{\OO_{1/2}} \enspace,\\
		g_1(\Delta) &= \{ \{ \YY, \ZZ \} \mid \YY, \ZZ \in f_1(\Delta), \textrm{ and there exist $Y \in \YY, Z \in \ZZ$ with } \{Y,Z\} \in g_{1/2}(\Delta) \} \enspace,\\ 
		h_1(\Delta) &= \{ \{ \YY, \ZZ \} \mid \YY, \ZZ \in f_1(\Delta), \{Y,Z\} \in h_{1/2}(\Delta) \textrm{ for all $Y \in \YY, Z \in \ZZ$ } \} \enspace.
	\end{align*}
	Now we transform $\Pi_1$ into a problem $\Pi_1^*$ that is at least as hard as $\Pi_1$, by removing outputs from $f_1(\Delta)$ while also (potentially) removing elements from $g_1(\Delta)$ and $h_1(\Delta)$.
	More specifically, our new output set $f_1^*(\Delta) = \OO_1^*$ contains exactly the sets $\YY \in f_1(\Delta)$ with the following properties:
	\begin{itemize}
		\item For each $Y \in \YY$, we have $|Y| = 2$. 
		\item For each set $Z \subset \{1,2,3,4\}$ with $|Z|=2$, exactly one of $Z$ and the complement of $Z$ (in $\{1,2,3,4\}$) is an element of $\YY$.
	\end{itemize} 
	We argue that this particular choice for $f_1^*(\Delta)$ satisfies two desirable properties, namely, that $\{ \YY, \ZZ \} \in g_1(\Delta)$ for any $\YY, \ZZ \in f_1^*(\Delta)$ with $\YY \neq \ZZ$, and $\{ \YY, \YY \} \in h_1(\Delta)$ for any $\YY \in f_1^*(\Delta)$.
	For the first property, observe that, for any $\YY, \ZZ \in f_1^*(\Delta)$ with $\YY \neq \ZZ$, there must be some set $Z \subset \{1,2,3,4\}$ with $|Z|=2$ such that $Z$ is contained in exactly one of $\YY$ and $\ZZ$.
	It follows that the other of the two contains the complement of $Z$, which in turn implies that there are elements $Y \in \YY, Z' \in \ZZ$ which are complementary (in $\{1,2,3,4\}$).
	We conclude that $\{Y, Z\} \in g_{1/2}(\Delta)$, which implies $\{ \YY, \ZZ \} \in g_1(\Delta)$.

	For the second property, observe that, for any $\YY \in f_1^*(\Delta)$ and any $Y, Y' \in \YY$, we know that $Y$ and $Y'$ are not complementary.
	Since both $Y$ and $Y'$ are sets of cardinality $2$, it follows that $Y \cap Y' \neq \emptyset$, which in turn implies that $\{Y, Y'\} \in h_{1/2}(\Delta)$.
	Hence, $\{ \YY, \YY \} \in h_1(\Delta)$.

	Let $g_1^*(\Delta)$ be obtained from $g_1(\Delta)$ by removing all contained multisets $\{ \YY, \ZZ \}$ with $\YY = \ZZ$.
	Similarly, let $h_1^*(\Delta)$ be obtained from $h_1(\Delta)$ by removing all contained multisets $\{ \YY, \ZZ \}$ with $\YY \neq \ZZ$.
	Let $k'$ denote the number of elements of $f_1^*(\Delta)$.
	Due to the two properties of $f_1^*(\Delta)$ we just proved and by renaming the elements of $f_1^*(\Delta)$, we can now describe $\Pi_1^*$ as follows.
	\begin{align*}
		f_1^*(\Delta) &= \OO_1^* = \{ 1, \dots, k' \} \enspace,\\
		g_1^*(\Delta) &= \{ \{ y, z \} \mid y, z \in \{ 1, \dots, k' \}, y \neq z \} \enspace,\\ 
		h_1^*(\Delta) &= \{ \{ y, y \} \mid 1 \leq y \leq k' \} \enspace.
	\end{align*}
	We conclude that $\Pi_1^*$ is the $k'$-coloring problem.

	Counting the number of elements of $f_1^*(\Delta)$, we obtain $k' = 2^{\binom{4}{2}/2}$.
	Analogously to the special case of $k=4$, we can derive the problem $\Pi_1^*$ for general even $k \geq 4$ (where $\Pi$ is defined as $k$-coloring).
	A close look at the general version of $f_1^*(\Delta)$ reveals that we then obtain $k' = 2^{\binom{k}{k/2}/2}$, which implies $k' \geq 2^{2^{k/2}}$ if $k \geq 6$.
	Applying our speedup theorems, we conclude that the existence of a $k'$-coloring algorithm implies the existence of a $k$-coloring algorithm that requires at most one additional round.
  	Since $2^{2^{k/2}}$ is a monotone function, this implies an (asymptotically) doubly exponential speedup per round, which in turn implies an $O(\log^* n)$ upper bound for $3$-coloring a ring.
	Note that, since our speedup theorems require $t$-independence, we will not assume that the given IDs are globally unique, but only that any two adjacent nodes have different IDs (which constitutes a coloring).
	However, since this assumption is weaker than assuming globally unique IDs and we are proving an upper bound, the bound directly applies to the $\LOCAL$ model.

	\subsection{Weak $2$-coloring}\label{sec:app2col}
	The goal of this section is to give an intuition for the generalization of weak $2$-coloring we use to prove our lower bound in Section\ \ref{sec:weak}.
	Recall that, in order to apply our speedup technique iteratively without having to look at each problem in the produced problem sequence individually, we would like to start with a problem $\Pi$ for which applying our speedup technique results in some problem that can be relaxed to a problem that is \emph{very similar} to $\Pi$.
	Ideally, the relaxed problem is the same as $\Pi$ except for that some parameters in the problem description are different.
	Also, our initial problem $\Pi$ should be similar to the problem we want to prove a lower bound for, i.e., weak $2$-coloring; the simplest case would be that $\Pi$ is a relaxation of weak $2$-coloring so that lower bounds for $\Pi$ immediately apply to weak $2$-coloring.
	
	In order to obtain some intuition for the effects of our speedup on problems that are similar to weak $2$-coloring, the obvious approach is to apply the speedup to weak $2$-coloring itself, which we do in the following.
	Recall that in the weak $2$-coloring (resp.\ weak $k$-coloring) problem, each node has to output a color from, say, the set $\{1,2\}$ (resp.\ $\{1, \dots, k\}$), such that there is at least one neighbor with a different color (or the node has no neighbors).
	Note that this formulation of the problem is not even edge-checkable; before we can apply our speedup we first have to bring weak $2$-coloring in a form that adheres to our definition of a problem.

	\paragraph{Problem $\Pi$}
	Consider the following problem $\Pi$:
	Each node has to output a color from $\{ 1, 2 \}$  and a pointer to an adjacent node.
	The output is correct if each node $v$ points to a node that has a different color than $v$.
	Formally, we can describe $\Pi$ by setting
	\begin{align*}
		f(\Delta) &:= \OO := \{ 1, 2 \} \times \{ \rightarrow, \bullet \} \enspace,\\ 
		g(\Delta) &:= \{ \{ (y,y'), (z,z') \} \mid y,z\in \{1,2\}, y',z'\in\{ \rightarrow, \bullet\}, y \neq z \veee y'=\bullet =z' \} \enspace,\\
		h(\Delta) &:= \{ \{ (y_1,y'_1), \dots, (y_{\Delta},y'_{\Delta}) \} \mid y_1,\dots,y_{\Delta}\in \{1,2\}, y_1 = \dots = y_{\Delta}, y'_1 =\hspace{0.4em} \rightarrow , y'_2 = \dots = y'_{\Delta} = \bullet \} \enspace.
	\end{align*}
	Note that, in the formal description, we use the fact that we are only interested in $\Delta$-regular graphs since we will prove our lower bound in Section\ \ref{sec:weak} already for this class of graphs.

	Clearly, $\Pi$ satisfies our definition of a problem; furthermore, we argue that if there is an algorithm solving weak $2$-coloring, then it can be transformed into an algorithm for $\Pi$ that requires only one additional round.
	This is easy to see: each node can learn the color of each adjacent node (according to the algorithm for weak $2$-coloring) in this additional round and then point to a node of different color.
	We call $\Pi$ the \emph{pointer version} of weak $2$-coloring; however, in this section we will simply refer to $\Pi$ as weak $2$-coloring (as it is essentially the same problem).

	\paragraph{Problem $\Pi'_{1/2}$}
	In the following, we derive problem $\Pi'_{1/2}$ from $\Pi$.
	According to the process of deriving $\Pi'_{1/2}$, described in Sections\ \ref{sec:thethe} and\ \ref{sec:maxisim}, we obtain $f_{1/2}(\Delta) = \OO_{1/2} = 2^{\{ 1, 2 \} \times \{ \rightarrow, \bullet \}}$.
	Furthermore, $h_{1/2}(\Delta)$ consists of all multisets $\{ Y_1, \dots, Y_{\Delta} \}$ such that A) each $Y_j$ is a subset of $\{ 1, 2 \} \times \{ \rightarrow, \bullet \}$, and B) there is some $c \in \{ 1, 2 \}$ such that $(c, \rightarrow) \in Y_1$ and $(c, \bullet) \in Y_i$ for all $2 \leq i \leq \Delta$.
	The most interesting part of the definition of $\Pi'_{1/2}$ is $g_{1/2}(\Delta)$ as Property\ \ref{prop:replacematch} comes into play here.
	Observe that for each multiset $\{Y,Z\}$ contained in $g_{1/2}(\Delta)$, it holds that if $(c,\rightarrow) \in Y$, then both $(c, \rightarrow)$ and $(c, \bullet)$ are not contained in $Z$.
	Taking the maximality constraint of Property\ \ref{prop:replacematch} into account, we obtain that $g_{1/2}(\Delta)$ contains exactly the multisets $\{ Y, Z \}$ where
	\begin{alignat*}{2}
		Y &= \{ (1,\rightarrow), (1,\bullet), (2,\rightarrow), (2,\bullet) \}, \qquad &&Z = \{ \}, \textrm{ or}\\
		Y &= \{ (1,\rightarrow), (1,\bullet) \}, &&Z = \{ (2,\rightarrow), (2,\bullet) \}, \textrm{ or}\\
		Y &= \{ (1,\rightarrow), (1,\bullet), (2,\bullet) \}, &&Z = \{ (2,\bullet) \}, \textrm{ or}\\
		Y &= \{ (2,\rightarrow), (2,\bullet), (1,\bullet) \}, &&Z = \{ (1,\bullet) \}, \textrm{ or}\\
		Y &= \{ (1,\bullet), (2,\bullet) \}, &&Z = \{ (1,\bullet), (2,\bullet) \}.
	\end{alignat*}

	\noindent Note that the definition of $g_{1/2}(\Delta)$ allows for one of $Y$ and $Z$ being empty; however, our function $h_{1/2}(\Delta)$ ensures that an empty set cannot be chosen as output, thereby ruling out the first listed row as a possibility for outputs $Y, Z$ on some edge.
	Hence, as we can see from the description of $g_{1/2}(\Delta)$, there are only $7$ outputs that can be used by any correct algorithm for $\Pi'_{1/2}$.
	This indicates that there may be a simpler, or more accessible, (equivalent) description of $\Pi'_{1/2}$, and indeed there (arguably) is, as given in the following.

	\paragraph{An Equivalent Description}
	Set $f_{1/2}(\Delta) := \OO_{1/2} := \{ 01, 02, 10, 11, 12, 20, 21 \}$, i.e., the set of allowed outputs is the set of all trit\footnote{A trit is the ternary analogue of a bit.} sequences of length $2$ excluding the two sequences $00$ and $22$.
	For any trit sequence $a=a_1\dots a_k$ of length $k$, any $0 \leq i \leq 2$, and any $1 \leq j \leq k$, we say that \emph{$a$ has an $i$ at position $j$} if $a_j = i$.
	Define $h_{1/2}(\Delta)$ to consist of all multisets of cardinality $\Delta$ of trit sequences from $f_{1/2}(\Delta)$ such that there is an index $1 \leq j \leq 2$ such that (at least) one of the trit sequences in the multiset has a $2$ at position $j$ and none of the trit sequences has a $0$ at position $j$.
	Furthermore, define $g_{1/2}(\Delta)$ to be the set of multisets $\{ a, a' \}$ of trit sequences from $f_{1/2}(\Delta)$ such that the tritwise addition of $a$ and $a'$ yields the trit sequence $22$.
	For instance, the multiset $\{ 02, 11, \dots, 11, 12, 21\}$ of cardinality $\Delta$ is an element of $h_{1/2}(\Delta)$ (pick $j=2$), and the multiset $\{ 01, 21 \}$ is an element of $g_{1/2}(\Delta)$.
	
	Why is this an equivalent description for $\Pi'_{1/2}$?
	Consider mapping the set of the ``usable" $7$ outputs in the initial description to our new output set $\OO_{1/2}$.
	More specifically, map each output $Y$ to the trit sequence $a_1a_2$ that satisfies that $a_j$ is equal to the number of elements from $\{ (j, \rightarrow), (j, \bullet) \}$ contained in $Y$. 
	For instance, $Y = \{ (2,\rightarrow), (2,\bullet), (1,\bullet) \}$ is mapped to the trit sequence $12$, whereas $Y = \{ (1,\rightarrow), (1,\bullet) \}$ is mapped to $20$.
	Clearly, this mapping is a bijection.
	Note that if $Y$ is mapped to a trit sequence $a_1a_2$ satisfying $a_j = 1$ for some $1 \leq j \leq 2$, then $(j, \bullet)$ is contained in $Y$, whereas $(j, \rightarrow)$ is not.
	Using this observation, it is straightforward to check that the two functions $h_{1/2}(\Delta)$ in the two descriptions are equivalent.
	In order to see that the same holds for $g_{1/2}(\Delta)$, it is sufficient to observe that in the initial description of $\Pi'_{1/2}$, for each multiset $(Y,Z)$ contained in $g_{1/2}(\Delta)$, the set $Y \cup Z$ contains exactly two elements where the first entry is $1$, and exactly two elements, where the first entry is $2$.
	Hence, our two definitions describe the same problem.
	We will use the second, arguably simpler description of $\Pi'_{1/2}$ for the next step on our agenda, i.e., for the task of deriving $\Pi'_1$ from $\Pi'_{1/2}$.
	
	\paragraph{Problem $\Pi'_1$}
	According to the derivation process for $\Pi'_1$, we obtain $f_1(\Delta) = \OO_1 = 2^{\OO_{1/2}}$.
	Furthermore, $g_1(\Delta)$ consists of all multisets $\{ W, X \}$ such that A) $W, X \subseteq \{ 01, 02, 10, 11, 12, 20, 21 \}$, and B) there exist elements $w \in W, x \in X$ such that the tritwise sum of $w$ and $x$ is $22$.
	Finally, $h_1(\Delta)$ consists of all multisets $\{ W_1, \dots, W_{\Delta} \}$ of sets $W_i \subseteq \{ 01, 02, 10, 11, 12, 20, 21 \}$  such that A) for any $w_1 \in W_1, \dots, w_{\Delta} \in W_{\Delta}$, there is an index $1 \leq j \leq 2$ such that (at least) one $w_i$ has a $2$ at position $j$ and none of the $w_i$ has a $0$ at position $j$, and B) the multiset is maximal with Property A), i.e., if we add a new element from $\{ 01, 02, 10, 11, 12, 20, 21 \}$ to some arbitrary $W_i$ in the multiset, then the multiset does not satisfy Property A) anymore.

	While the description of $g_1(\Delta)$ provides a good intuition for which $2$-element multisets are actually contained in $g_1(\Delta)$, the maximality condition for $h_1(\Delta)$ makes it harder to see the same for $h_1(\Delta)$.
	In fact, $h_1(\Delta)$ is an excellent example for the power of the simplification technique introduced in Section\ \ref{sec:maxisim}: with some work, it is possible to show that $h_1(\Delta)$ actually contains only $9$ elements (or fewer if $\Delta$ is very small)!
	As the goal of this section is merely to provide intuition for our lower bound approach in Section\ \ref{sec:weak}, we will only highlight the important observations that can be made by examining those $9$ elements.

	\paragraph{Problem $\Pi^*_1$}
	As outlined in the beginning of this section, we would like to relax $\Pi'_1$ to a problem $\Pi^*_1$ that is similar to weak coloring and has a simpler description.
	A very natural idea is to try to relax $\Pi'_1$ to weak $9$-coloring as follows.
	Map each of the $9$ elements of $h_1(\Delta)$ to a different color and then show that any algorithm $\AAA$ for $\Pi'_1$ also solves weak $9$-coloring in the following way: Each node $v$ executes $\AAA$ to determine the outputs at all $(v,e)$, then checks which element of $h_1(\Delta)$ the outputs constitute together, and finally outputs the color corresponding to the element of $h_1(\Delta)$ at each $(v,e)$.
	On top of that, each node (somehow) infers from its output according to $\AAA$ an adjacent node it can safely point to (i.e. an adjacent node that must have a different color).
	
	While for most of the $9$ elements in $h_1(\Delta)$ this approach can be made to work\footnote{In fact, using several speedup-related tricks, it is possible to relax $\Pi'_1$ to to weak $11$-coloring; however, these tricks seem to apply only for our very specific case of $\Pi$ being weak $2$-coloring and do not generalize to weak $k$-coloring or similar problems.}, there is one element in $h_1(\Delta)$ that some node $v$ and \emph{all} its neighbors might ``output" (if we gather all partial outputs of the respective node according to algorithm $\AAA$).
	In other words, the envisioned solution of $\AAA$ for weak $9$-coloring can make $v$ and all its neighbors output the same color, which makes it impossible for $v$ to output a correct pointer.
	However, this particular element of $h_1(\Delta)$ (recall that such an element is a multiset of cardinality $\Delta$) has a very special form: we can write it as $Q = \{ Q_1, Q_2, Q_3, Q_4, \dots, Q_4 \}$, where $\{Q_1, Q_3\}$ and $\{Q_2, Q_3\}$ are the only two elements of $g_1(\Delta)$ that involve $Q_1$ or $Q_2$ (i.e., if the outputs of two adjacent nodes $u, v$ according to $\AAA$ both correspond to $Q$ and the output at $(u,e)$ is $Q_1$ or $Q_2$, then the output at $(v,e)$ must be $Q_3$, where $e=\{u,v\}$).
	
	This leads to the idea of modifying our weak $9$-coloring problem slightly: Instead of having to output one pointer that has to point to a differently colored node, a node $v$ can also choose to output two pointers and one ``inverse pointer" (at three different $(v,e)$), where a (non-inverse) pointer has to point to a differently colored node or to an inverse pointer.
	For this modified problem, any node $v$ has a simple way to infer a correct output if its output according to $\AAA$ corresponds to $Q$: node $v$ simply outputs a pointer at the two $(v,e)$ where $\AAA$ outputs $Q_1$ or $Q_2$, and an inverse pointer at the one $(v,e)$ where $\AAA$ outputs $Q_3$.
	
	As we will see (in Lemma\ \ref{lem:hall}), this particular property of $Q$ is not an isolated case but \emph{characteristic} for problems $\Pi'_1$ that are derived from the problems we generalize weak $2$-coloring to.
	More precisely, for any such $\Pi'_1$ and any $Q$ in the corresponding $h_1(\Delta)$, we can essentially write $Q$ as $\{ Q_1, \dots Q_i, Q'_1, \dots, Q'_j, Q''_1, \dots, Q''_{\ell}\}$ such that 1) $i > j$, and 2) if two nodes $u,v$ connected by an edge $e$ both output $Q$, and $u$ outputs some element from $\{ Q_1, \dots Q_i\}$ at $(u,e)$, then $v$ has to output some element from $\{ Q'_1, \dots, Q'_j \}$ at $(v,e)$ (by the definition of the corresponding $g_1(\Delta)$).
	This insight suggests the generalization of weak $2$-coloring that we present in Section\ \ref{sec:weak}.

	\section{A Tight Lower Bound for Odd-Degree Weak $2$-Coloring}\label{sec:weak}
	This section is devoted to proving a lower bound for weak $2$-coloring on odd-degree graphs that matches the upper bound presented by Naor and Stockmeyer \cite{DBLP:conf/stoc/NaorS93}.
	In Section\ \ref{sec:speele}, we introduce the problem of superweak $k$-coloring, which is a generalization of weak $2$-coloring, and prove a speedup lemma for this new problem.
	More precisely, we show that we can save one round in the runtime if we increase the parameter $k$ in the problem definition exponentially (a constant number of times).
	Subsequently, in Section\ \ref{sec:lowerbound}, we will use the speedup lemma to prove our lower bound.

	\subsection{A Speedup Lemma for Superweak $k$-Coloring}\label{sec:speele}
	We start this section by defining superweak $k$-coloring and applying our speedup to it, resulting in a problem $\Pi'_1$.
	By showing (Lemma\ \ref{lem:zeroround}) that, for some sufficiently large $k' > k$, superweak $k'$-coloring is a relaxed version of $\Pi'_1$ (i.e., superweak $k'$-coloring can be solved at least as fast as $\Pi'_1$), we can infer our speedup lemma (Lemma\ \ref{lem:supersuper}).

	Our main technical ingredient to prove Lemma\ \ref{lem:zeroround} is a statement (Lemma\ \ref{lem:hall}) that, roughly speaking, ensures\footnote{Note that this is essentially a reformulation of the insight obtained at the very end of Section\ \ref{sec:app2col}.} that for each output of a node $v$ executing $\Pi'_1$ (i.e., for each collection of the partial outputs at the $(v,e)$ for a fixed node $v$), we can select two disjoint sets $E'_v$ and $E''_v$ of edges incident to $v$ such that $|E'_v| > |E''_v|$ and the following property is satisfied:
	If two adjacent nodes $u, v$ give the same output and the connecting edge $e$ is contained in $E'_u$, then $e$ must also be contained in $E''_v$.
	The edges in $E'_v$, resp.\ $E''_v$, will be exactly those edges where $v$ outputs the demanding, resp.\ accepting, pointers specified in the definition of superweak $k$-coloring.
	
	In order to prove Lemmas\ \ref{lem:hall} and\ \ref{lem:zeroround}, we require the existence of a structural property of the elements in the set $h_1(\Delta)$ (defining $\Pi'_1$), which is provided by Lemma\ \ref{lem:infty}.
	In particular, this property is responsible for bounding the increase from $k$ to $k'$.

	\paragraph{Superweak $k$-Coloring}
	Let $k \geq 2$ be some integer.
	We define the problem of superweak $k$-coloring as follows (using the insights obtained in Section\ \ref{sec:app2col}).
	Each node has to output a color from $\{ 1, \dots, k \}$  and a number of pointers to adjacent nodes.
	There are two kinds of pointers, \emph{demanding} pointers and \emph{accepting} pointers.
	A node is not allowed to point to the same neighbor with two pointers; furthermore, the number of demanding pointers a node uses must be strictly greater than the number of accepting pointers it uses (the latter can be $0$).
	The number of accepting pointers a node is allowed to use is bounded from above by $k$. 
	The output is correct if for each demanding pointer from a node $v$ to a node $u$, the two nodes have different colors or $u$ has an accepting pointer pointing to $v$.
	For an illustration, see Figure \ref{fig:super}.
	\begin{figure}
		\centering
		\includegraphics[scale=1.1]{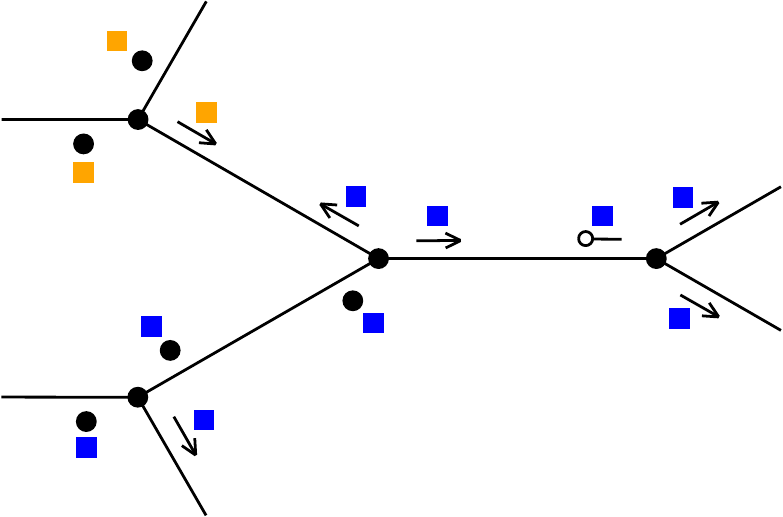}
		\caption{A locally correct output for superweak $k$-coloring, where $k \geq 2$ and $\Delta = 3$: Each node outputs the same color on all incident edges and strictly more demanding pointers than accepting pointers, and the number of accepting pointers per node is at most $k$. On each edge, we have different colors, no demanding pointer, or a demanding pointer and an accepting pointer.}
		\label{fig:super}
	\end{figure}
	
	Formally, we define superweak $k$-coloring as follows, again using that we are only interested in regular graphs (since our lower bound will hold already for regular graphs).
	Set
	\begin{align*}
		f(\Delta) &:= \OO := \{ 1, \dots, k \} \times \{ \rightarrow, \multimap, \bullet \} \enspace,\\ 
		g(\Delta) &:= \{ \{ (y, y'), (z, z') \} \mid y,z\in \{1, \dots, k\}, y',z'\in\{ \rightarrow, \multimap, \bullet\}, y \neq z \veee y'=\bullet =z' \veee \multimap \hspace{0.3em} \in \{ y', z' \} \} \enspace,\\
		h(\Delta) &:= \{ \{ (y_1,y'_1), \dots, (y_{\Delta},y'_{\Delta}) \} \mid y_1,\dots,y_{\Delta}\in \{1, \dots, k\}, y'_1,\dots,y'_{\Delta}\in\{ \rightarrow, \multimap, \bullet\}, y_1 = \dots = y_{\Delta},\\
		& \qquad \min \{ k+1, |\{ 1 \leq i \leq \Delta \mid y_i = \hspace{0.3em} \rightarrow \}| \} > |\{ 1 \leq i \leq \Delta \mid y_i = \hspace{0.3em} \multimap \}| \hspace{0.2em} \} \enspace.
	\end{align*}

	\paragraph{Problem $\Pi'_{1/2}$}
	Now we examine how the derived problem $\Pi'_{1/2}$ looks like, if we define $\Pi$ to be superweak $k$-coloring.
	We obtain $f_{1/2}(\Delta) = \OO_{1/2} = 2^{\{ 1, \dots, k \} \times \{ \rightarrow, \multimap, \bullet \}}$.
	Furthermore, $h_{1/2}(\Delta)$ consists of all multisets $\{ Y_1, \dots, Y_{\Delta} \}$ such that A) each $Y_j$ is a subset of $\{ 1, \dots, k \} \times \{ \rightarrow, \multimap, \bullet \}$, and B) there are some $c \in \{ 1, \dots, k \}$ and some partition $U_1 \mathop{\dot{\cup}} U_2 \mathop{\dot{\cup}} U_3 = \{ 1, \dots, \Delta \}$ such that $(c, \rightarrow) \in Y_i$ if $i \in U_1$, $(c, \multimap) \in Y_i$ if $i \in U_2$, $(c, \bullet) \in Y_i$ if $i \in U_3$, $|U_1| > |U_2|$, and $|U_2| \leq k$.
	Similar to the case of $\Pi'_{1/2}$ being derived from weak $2$-coloring, observe that for each multiset $\{Y,Z\}$ contained in $g_{1/2}(\Delta)$, it holds that if $(c,\rightarrow) \in Y$, then both $(c, \rightarrow)$ and $(c, \bullet)$ are not contained in $Z$, whereas $(c, \multimap)$ is always contained in both $Y$ and $Z$ for all $1 \leq c \leq k$.
	Again taking the maximality constraint of Property\ \ref{prop:replacematch} into account, we obtain that $g_{1/2}(\Delta)$ contains exactly the multisets $\{ Y, Z \}$ with the property that, for each $1 \leq c \leq k$, we have
	\begin{alignat*}{2}
		Y^c &= \{ (c,\rightarrow), (c, \multimap) (c,\bullet) \}, \qquad &&Z^c = \{ (c, \multimap) \}, \textrm{ or}\\
		Y^c &= \{ (c, \multimap), (c,\bullet) \}, &&Z^c = \{ (c, \multimap), (c,\bullet) \},
	\end{alignat*}
	where $Y^c$ and $Z^c$ denote the subset of $Y$ and $Z$, respectively, consisting of all elements where the first entry is $c$.
	
	\paragraph{An Equivalent Description}
	Analogously to our approach in the case of $\Pi$ being weak $2$-coloring, we can give a different description of the same problem:
	Define $f_{1/2}(\Delta) = \OO_{1/2}$ to be the set of all trit sequences of length $k$.
	Define $h_{1/2}(\Delta)$ to consist of all multisets of cardinality $\Delta$ of trit sequences of length $k$ such that there is an index $1 \leq j \leq k$ such that, in the multiset, the number of trit sequences that have a $2$ at position $j$ is strictly greater than the number of trit sequences that have a $0$ at position $j$, and there are at most $k$ trit sequences that have a $0$ at position $j$.
	Finally, define $g_{1/2}(\Delta)$ to be the set of multisets $\{ a, a' \}$ of trit sequences of length $k$ such that the tritwise addition of $a$ and $a'$ yields the trit sequence $22\dots 2$.

	In order to see that this new description is equivalent to the old one, we can use the same argumentation as in the analogous case: essentially, we map $\{ (c, \multimap) \}$ to a $0$ at position $c$, $\{ (c, \multimap), (c,\bullet) \}$ to a $1$, and $\{ (c,\rightarrow), (c, \multimap),   (c,\bullet) \}$ to a $2$, and then observe that the constraints are the same in both descriptions.
	
	\paragraph{Problem $\Pi'_1$}
	Unavoidably, the next step is to derive $\Pi'_1$ from $\Pi'_{1/2}$.
	We obtain $f_1(\Delta) = \OO_1 = 2^{\OO_{1/2}}$.
	Furthermore, $g_1(\Delta)$ consists of all multisets $\{ W, X \}$ such that A) $W, X$ are sets of trit sequences of length $k$, and B) there exist elements $w \in W, x \in X$ such that the tritwise sum of $w$ and $x$ is $22\dots 2$.
	Finally, $h_1(\Delta)$ consists of all multisets $\{ W_1, \dots, W_{\Delta} \}$ of sets $W_i$ of trit sequences of length $k$ such that A) for any $w_1 \in W_1, \dots, w_{\Delta} \in W_{\Delta}$, there is an index $1 \leq j \leq k$ such that the number of $w_i$ that have a $2$ at position $j$ is strictly greater than the number of $w_i$ that have a $0$ at position $j$, and there are at most $k$ many $w_i$ that have a $0$ at position $j$, and B) the multiset is maximal with Property A), i.e., if we add a new trit sequence of length $k$ to some arbitrary $W_i$ in the multiset, then the multiset does not satisfy Property A) anymore.
	
	The following lemma highlights a structural aspect of the elements in $h_1(\Delta)$.
	
	\begin{lemma}\label{lem:infty}
		Let $Q = \{ Q_1, \dots, Q_{\Delta} \}$ be an arbitrary multiset from $h_1(\Delta)$, and assume that $\Delta \geq 2^{4^k+1}$.
		Then there is a unique element $P_{\infty} \in Q$ that has multiplicity at least $\Delta - 2^{4^k}$ and contains the trit sequence $11\dots 1$.
	\end{lemma}
	
	\begin{proof}
		Let $\{ P_1, \dots, P_{i'} \}$ be the set of all elements (i.e., of all sets of trit sequences of length $k$) of $Q$ that occur in $Q$ with multiplicity at least $(k+1) \cdot 3^k$.
		Set $P := \bigcup_{i=1}^{i'} P_i$.
		Let $Q' = \{Q'_1, \dots, Q'_{\Delta}\}$ be the multiset obtained by replacing in $Q$ all occurrences of $P_i$, except exactly $(k+1) \cdot 3^k$ many, by $P$, for each $1 \leq i \leq i'$.
		Then the multiplicity of each $P_i$ in $Q'$ is exactly $(k+1) \cdot 3^k$.
		We claim that $Q' \in h_1(\Delta)$.

		For a contradiction, assume that $Q' \notin h_1(\Delta)$.
		Since $Q \in h_1(\Delta)$ (and $P \supseteq P_i$ for any $i$), it is not possible that $Q'$ violates \emph{only} Property B) in the description of $h_1(\Delta)$.
		Hence, $Q'$ violates Property A) in that description, which implies that there are $q'_1 \in Q'_1, \dots, q'_{\Delta} \in Q'_{\Delta}$ such that, for each index $1 \leq j \leq k$, the number of $q'_{\ell}$ (where $1 \leq \ell \leq \Delta$) that have a $0$ at position $j$ is greater than $k$ or at least as large as the number of $q'_{\ell}$ that have a $2$ at position $j$.
		
		For each $P_i$, let $p_i \in P_i$ be a trit sequence with multiplicity at least $k+1$ in the multiset $\{ q'_1, \dots, q'_{\Delta} \}$.
		Since each $P_i$ has multiplicity $(k+1) \cdot 3^k$ in $Q'$ (and the number of distinct trit sequences of length $k$ is $3^k$), such $p_i$ exist, by the pidgeonhole principle.
		Now consider trit sequences $q_1 \in Q_1, \dots, q_{\Delta} \in Q_{\Delta}$ specified as follows:

		First, for each index $\ell$ that satisfies $Q_{\ell} = Q'_{\ell}$ and $Q_{\ell} \neq P_i$ for all $1 \leq i \leq i'$, set $q_{\ell} := q'_{\ell}$.
		Then, for each trit sequence $p \in P$ that has multiplicity $m_p \leq k$ in $\{ q'_1, \dots, q'_{\Delta} \}$, first compute $m'_p := m_p - p^*$ where $p^*$ is the number of $q_{\ell}$ already set to $p$ in the first step, then pick both an arbitrary $P_i$ containing $p$ and $m'_p$ many so far unused indices $\ell$ with $Q_{\ell} = P_i$, and set $q_{\ell} := p$ for these indices.
		Then, for each trit sequence $p \in P$ that has multiplicity $m_p \geq k+1$ in $\{ q'_1, \dots, q'_{\Delta} \}$, pick both an arbitrary $P_i$ containing $p$ and $k+1$ many so far unused indices $\ell$ with $Q_{\ell} = P_i$, and set $q_{\ell} := p$ for these indices.
		Finally, for each remaining unused index $\ell$, which must necessarily satisfy $Q_{\ell} = P_i$ for some $1 \leq i \leq i'$, set $q_{\ell} := p_i$ (where $p_i$ is as specified above).

		Since each $P_i$ has multiplicity $(k+1) \cdot 3^k$ in $Q'$, such $q_1, \dots, q_{\Delta}$ exist, again by the pidgeonhole principle.
		Consider the two multisets $\{ q_1, \dots, q_{\Delta} \}$ and $\{ q'_1, \dots, q'_{\Delta} \}$.
		Observe that each trit sequence of length $k$ has either multiplicity at least $k+1$ in both multisets, or the same multiplicity $\leq k$ in both multisets.
		Hence, for each $1 \leq j \leq k$, the number of $q_{\ell}$ that have a $0$ at position $j$ is the same as the number of $q'_{\ell}$ that have a $0$ at position $j$, or both numbers are at least $k+1$.
		For analogous reasons, the statement still holds if we replace each $0$ by a $2$.
		Taking a close look at Property A), we conclude that if Property A) fails for $Q'$ due to choosing $q'_1 \in Q'_1, \dots, q'_{\Delta} \in Q'_{\Delta}$, then Property A) also fails for $Q$ due to choosing $q_1 \in Q_1, \dots, q_{\Delta} \in Q_{\Delta}$.
		Since we already established the former, we obtain a contradiction to $Q \in h_1(\Delta)$.
		Thus, $Q' \in h_1(\Delta)$.
		
		As $Q'_{\ell} \supseteq Q_{\ell}$ for each $1 \leq \ell \leq \Delta$ and $Q, Q' \in h_1(\Delta)$, we see that $Q = Q'$, by Property B) in the description of $h_1(\Delta)$.
		Hence, there is at most one element in $Q$ with multiplicity greater than $(k+1) \cdot 3^k$.
		Since $Q$ contains at most $2^{3^k}$ distinct elements, it follows that there is a unique element $P_{\infty} \in Q$ with multiplicity at least\footnote{Note that, for the sake of simplicity, we do not attempt to prove the sharpest bounds possible, but rather bounds that are convenient to work with.} $\Delta - 2^{4^k} \leq \Delta - (k+1) \cdot 3^k \cdot 2^{3^k}$, due to $k \geq 2$ and our assumption that $\Delta \geq 2^{4^k+1}$.

		In order to see that $P_{\infty}$ contains $11\dots 1$, assume for a contradiction that it does not, and define $Q'' = \{Q''_1, \dots, Q''_{\Delta}\}$ to be the multiset obtained by adding the trit sequence $11\dots 1$ to one of the $P_{\infty}$ contained in $Q$, say to $Q_{\ell} = P_{\infty}$.
		Analogously to the above, we can infer that $Q''$ violates Property A) and consider a violating choice $q''_1 \in Q''_1, \dots, q''_{\Delta} \in Q''_{\Delta}$.
		We must have $q''_{\ell} = 11\dots 1$ in that violating choice; if we now replace $q''_{\ell}$ by some trit sequence contained in $P_{\infty}$ that has multiplicity at least $k+1$ in the multiset $\{q''_1, \dots, q''_{\Delta}\}$ (which must exist due to the pidgeonhole principle), we obtain a choice $q''_1 \in Q_1, \dots, q''_{\Delta} \in Q_{\Delta}$ that ensures that $Q$ violates Property A), as can be shown using analogous arguments as above.
		This, again, yields a contradiction to $Q \in h_1(\Delta)$, thereby proving that $P_{\infty}$ contains $11\dots 1$.
	\end{proof}
	
	Note that the element $P_{\infty}$ in Lemma\ \ref{lem:infty} depends on the choice of $Q$.
	If we want to specify the respective $Q$, we write $P^Q_{\infty}$ instead of $P_{\infty}$.

	The following technical lemma provides us with two index sets $J^*$ and $N(J^*)$ (depending on the local output of an algorithm for $\Pi'_1$) with desirable properties that we will use in Lemma\ \ref{lem:zeroround} to infer where to output demanding and accepting pointers (when transforming an algorithm for $\Pi'_1$ into an algorithm for superweak $k'$-coloring).
	Curiously, our proof of Lemma\ \ref{lem:hall} requires the application of Hall's marriage theorem \cite{hall}.
	
	\begin{lemma}\label{lem:hall}
		Let $Q = \{ Q_1, \dots, Q_{\Delta} \}$ be an arbitrary multiset from $h_1(\Delta)$, and assume that $\Delta \geq 2^{4^k+1}$.
		Let $\alpha: \{ 1, \dots, \Delta\} \rightarrow \{ \out, \ing\}$ be a function that maps each $1 \leq i \leq \Delta$ to an element of $\{ \out, \ing \}$.
		Let $\II \subseteq \{ 1, \dots, \Delta\}$ be the subset of indices $i$ satisfying $\{ Q_i, P_{\infty} \} \notin g_1(\Delta)$ and\footnote{The statement that $\{ Q_i, P_{\infty} \} \notin g_1(\Delta)$ already implies that $11\dots 1 \notin Q_i$, by Lemma\ \ref{lem:infty}, but let us be very explicit here, for convenience.} $11\dots 1 \notin Q_i$.
		For any $\JJ \subseteq \II$, let $N(\JJ) \subseteq \{ 1, \dots, \Delta\}$ denote the subset of indices $i$ such that there exists an index $j \in \JJ$ such that $\{ Q_i, Q_j \} \in g_1(\Delta)$ and $\alpha(i) \neq \alpha(j)$.
		Then there is some $\JJ^* \subseteq \II$ such that $|\JJ^*| > |N(\JJ^*)|$ and for any $j \in \JJ^*, i \in N(\JJ^*)$, we have $\alpha(j) \neq \alpha(i)$.
	\end{lemma}

	\begin{proof}
		Consider the bipartite (unlabeled) graph $G'$ defined as follows.
		One side of the bipartition is given by a node $v_i$ for each index $i \in \II$; the other side by a node $u_j$ for each index $1 \leq j \leq \Delta$.
		Deviating from our conventions, we denote the set of $v_i$ by $V(G')$ and the set of $u_j$ by $U(G')$; furthermore, for each subset $V' \subseteq V(G')$ and each subset $U' \subseteq U(G')$, we set $I(V') := \{ i \in \II \mid v_i \in V' \}$ and $I(U') := \{ j \in \{ 1, \dots, \Delta \} \mid u_j \in U' \}$, respectively.
		The edge set $E(G')$ is given by creating an edge between $v_i$ and $u_j$ if and only if $\{ Q_i, Q_j \} \in g_1(\Delta)$ and $\alpha(i) \neq \alpha(j)$.
		In other words, $U' \subseteq U(G')$ is the set of neighbors of some set $V' \subseteq V(G')$ in $G'$ if and only if $I(U') = N(I(V'))$.
		We denote the set of neighbors of a set $V' \subseteq V(G')$, resp.\ $U' \subseteq U(G')$, in $G'$ by $N(V')$, resp.\ $N(U')$.
		
		We start proving the lemma by showing that there exists some set $\JJ' \subseteq \II$ such that $|\JJ'| > |N(\JJ')|$.
		Assume for a contradiction that no such set $\JJ'$ exists, i.e., for each subset $\JJ \subseteq \II$, we have $|\JJ| \leq |N(\JJ)|$.
		This implies that for each subset $V' \subseteq V(G')$, we have $|V'| \leq |N(V')|$.
		Now we can apply Hall's marriage theorem \cite[Theorem 1]{hall} that states that in bipartite graphs $G'$ with the bipartition given by the two node sets $V(G')$ and $U(G')$, the above condition is sufficient (and also necessary) for the existence of a matching that covers all nodes in $V(G')$.
		Let $M \subseteq E(G')$ be such a matching.
		Hence, for each $v \in V(G')$, there is an edge $e \in M$ with $v \in e$.
		Note that no node $v_i$ can be matched to $u_i$ in $M$ since there is no edge in $E(G')$ connecting $v_i$ to $u_i$. 
		We say that two edges $e,e' \in M, e \neq e',$ are \emph{touching} if there is some $i \in \II$ such that both $e$ and $e'$ have an endpoint in $\{ v_i, u_i \}$.

		Our next goal is to find a matching $M' \subseteq M$ with the property that for each $i \in \II$, exactly one of $v_i$ and $u_i$ is matched.
		To this end, find a partition $M = M_1 \mathop{\dot{\cup}} \dots \mathop{\dot{\cup}} M_{r}$ such that two edges $e, e' \in M$ are contained in the same $M_j$ if and only if there exists a sequence of edges $e = e_1, e_2, \dots, e_s = e'$ such that $e_i$ and $e_{i+1}$ are touching for all $1 \leq i \leq s-1$.
		Clearly, such a partition exists, and it is easy to see that this partition is unique up to reordering of the $M_j$.
		Observe that the definition of touching edges implies that for each edge $e \in M$, there are at most two other edges $e' \in M$ such that $e$ and $e'$ are touching.
		It follows, due to the definition of our partition and $G'$, that each of the $M_j$ has exactly one of the two following forms:
		\begin{enumerate}
			\item\label{prop:path} There is a sequence $i_1, \dots, i_s$ of distinct indices from $\{1, \dots, \Delta\}$ such that $M_j = \{ \{ v_{i_{\ell}}, u_{i_{\ell+1}} \} \mid 1 \leq \ell \leq s-1 \}$, $i_s \notin \II$, and node $u_{i_1}$ is not matched in $M$.
			\item\label{prop:ring} There is a sequence $i_1, \dots, i_s$ of distinct indices from $\{1, \dots, \Delta\}$ such that $M_j = \{ \{ v_{i_{\ell}}, u_{i_{\ell+1}} \} \mid 1 \leq \ell \leq s-1 \} \cup \{ \{ v_{i_s}, u_{i_1} \} \}$.
		\end{enumerate}
		We say that $M_j$ that fall into the first category are \emph{pathlike}, and that $M_j$ from the second category are \emph{ringlike}.
		Note that for pathlike $M_j$, the choice of the sequence of indices is unique, whereas for ringlike $M_j$, the choice is only unique up to certain reorderings.
		In the latter case, fix an arbitrary one of those sequences.
		We observe that for ringlike $M_j$, the parameter $s$ is even:
		As the edges contained in $M_j$ are also elements of $E(G')$, we have $\alpha(i_1) \neq \alpha(i_2) \neq \dots \neq \alpha(i_s) \neq \alpha(i_1) $, by the definitions of $M_j$ and $E(G')$.
		Since $\alpha(i_1), \dots, \alpha(i_s) \in \{ \out, \ing \}$, we obtain $\alpha(i_1) = \alpha(i_3) = \dots$, which implies that $s$ must be even.

		In order to find our desired matching $M'$, we will describe for each $M_j$ which of the edges contained in $M_j$ are chosen as elements of $M'$.
		For both pathlike and ringlike $M_j$, choose the edges $\{ v_{i_{\ell}}, u_{i_{\ell+1}} \}$ where $\ell$ is odd into $M'$.
		We argue that the obtained matching $M'$ has indeed the property that for each $i \in \II$, exactly one of $v_i$ and $u_i$ is matched:
		Consider an arbitrary $i \in \II$.
		Either both $v_i$ and $u_i$ are matched in our initial matching $M$, in which case the two edges from $M$ responsible for this are in the same $M_j$ and exactly one of the two edges is chosen into $M'$; or only $v_i$ is matched in $M$, in which case the edge from $M$ responsible for this must be the edge $\{ v_{i_1}, u_{i_2} \}$ in some pathlike $M_j$ and this edge is chosen into $M'$ whereas $u_i$ remains unmatched also in $M'$.

		Now we will use our matching $M'$ to find a choice $q_1 \in Q_1, \dots, q_{\Delta} \in Q_{\Delta}$ that violates Property A) in the description of $h_1(\Delta)$.
		More specifically, we choose our $q_{\ell}$ as follows.
		
		First, for each edge $\{ v_i, u_j \} \in M'$, choose $q_i, q_j$ such that the tritwise sum of $q_1$ and $q_2$ is $22\dots 2$.
		This is possible, as $\{ Q_i, Q_j \} \in g_1(\Delta)$, due to the definition of $E(G')$ and the fact that $M' \subseteq E(G')$.
		Moreover, since $M'$ is a matching with the property described above, the above rule fixes each $q_{\ell}$ at most once, and it fixes $q_{\ell}$ for all $\ell \in \II$.
		Second, consider all $\ell \notin \II$ that satisfy $11\dots 1 \notin Q_{\ell}$.
		Due to the definition of $\II$, all these $\ell$ satisfy $\{ Q_{\ell}, P_{\infty} \} \in g_1(\Delta)$.
		For each such $\ell$, find a partner index $j_{\ell}$ such that $Q_{j_{\ell}} = P_{\infty}$, and all these $j_{\ell}$ are distinct and have not been used in the above (i.e., $q_{j_{\ell}}$ is not fixed yet).
		Since, by Lemma\ \ref{lem:infty}, $P_{\infty}$ has a multiplicity of at least $\Delta/2$ in $Q$, this is possible. 
		Now, for each such $\ell$, choose $q_{\ell}, q_{j_{\ell}}$ such that the tritwise sum of $q_{\ell}$ and $q_{j_{\ell}}$ is $22\dots 2$ (which, again, is possible due to $\{ Q_{\ell}, P_{\infty} \} \in g_1(\Delta)$).
		Finally, for each $\ell$ for which $q_{\ell}$ is not fixed yet, set $q_{\ell} := 11\dots 1$ (which is possible as all leftover $\ell$ satisfy $11\dots 1 \in Q_{\ell}$).

		Since each $q_{\ell}$ is equal to $11\dots 1$ or has a unique partner such that the tritwise sum of the two is $22\dots 2$, the above choice of the $q_{\ell}$ ensures that, for each $1 \leq j \leq k$, the number of $q_{\ell}$ that have a $0$ at position $j$ is exactly the same as the number of $q_{\ell}$ that have a $2$ at position $j$.
		By the definition of $h_1(\Delta)$, it follows that $Q \notin h_1(\Delta)$, yielding a contradiction.
		Hence, there exists some set $\JJ' \subseteq \II$ such that $|\JJ'| > |N(\JJ')|$.

		In order to use $\JJ'$ to obtain a set $\JJ^* \subseteq \II$ as specified in the lemma, partition $\JJ'$ into two sets $\JJ'_{\ing}$ and $\JJ'_{\out}$, where $\JJ'_{\ing} := \{ j \in \JJ' \mid \alpha(j)=\ing \}$ and $\JJ'_{\out} := \{ j \in \JJ' \mid \alpha(j)=\out \}$.
		Consider $N(\JJ'_{\ing})$ and $N(\JJ'_{\out})$.
		By the definition of the function $N()$, each $i \in N(\JJ'_{\ing})$ has to satisfy $\alpha(i) = \out$, and each $i \in N(\JJ'_{\out})$ has to satisfy $\alpha(i) = \ing$.
		Since $N(\JJ') = N(\JJ'_{\ing}) \cup N(\JJ'_{\out})$, it follows that $N(\JJ')$ is actually the \emph{disjoint} union of $N(\JJ'_{\ing})$ and $N(\JJ'_{\out})$.
		Hence, we have $|\JJ'_{\ing}| + |\JJ'_{\out}| = |\JJ'| > |N(\JJ')| = |N(\JJ'_{\ing})| + |N(\JJ'_{\out})|$, which implies that $|\JJ'_{\ing}| > |N(\JJ'_{\ing})|$ or $|\JJ'_{\out}| > |N(\JJ'_{\out})|$.
		
		Now, if $|\JJ'_{\ing}| > |N(\JJ'_{\ing})|$, set $\JJ^* := \JJ'_{\ing}$, otherwise set $\JJ^* := \JJ'_{\out}$.
		By our above observations, we conclude that $\JJ^*$ satisfies the properties stated in the lemma.
	\end{proof}

	The following lemma shows that superweak $k'$-coloring is a relaxation of our problem $\Pi'_1$ derived from superweak $k$-coloring, if $k'$ is sufficiently large compared to $k$.
	The proof relies heavily on Lemma\ \ref{lem:hall}.

	\begin{lemma}\label{lem:zeroround}
		Let $\Pi$ be superweak $k$-coloring for some $k \geq 2$, and let $t$ be an arbitrary non-negative integer.
		Let $\GG$ be a graph class such that every contained graph is regular with (maximum) degree at least $2^{4^k+1}$ and the given input includes an orientation for each edge.
		Set $k' := 2^{2^{5^k}}$.
		If there is an algorithm that solves $(\Pi'_1, \GG)$ in time $t$, then there is also an algorithm that solves superweak $k'$-coloring on $\GG$ in time $t$.
	\end{lemma}
	
	\begin{proof}
		Let $\AAA$ be an algorithm that solves $(\Pi'_1, \GG)$ in time $t$.
		We will construct an algorithm $\AAA'$ that solves superweak $k'$-coloring on $\GG$ in time $t$ as follows.
		For executing $\AAA'$ each node $v$ first executes $\AAA$ and then applies a function that takes the obtained outputs at all $(v,e)$, where $e$ is incident to $v$, and the input information and port numbers $v$ knows from the very beginning as input and returns new outputs for all the $(v,e)$.
		As applying the function requires no communication between the nodes, $\AAA'$ has the same runtime as $\AAA$.
		In the following we describe the function, i.e., how $v$ transforms an arbitrary output for $\Pi'_1$ (at all $(v,e)$) into an output for superweak $k'$-coloring.

		Let $G \in \GG$ and $v \in V(G)$, and denote the maximum degree of $G$ by $\Delta$.
		Let $Q^v = \{ Q_1, \dots, Q_{\Delta} \}$ denote the multiset that contains, for each edge $e$ incident to $v$, the output of $\AAA$ at $(v,e)$.
		Since $\AAA$ solves $(\Pi'_1, \GG)$, we know that $Q^v \in h_1(\Delta)$.
		Let $e_1, \dots, e_{\Delta}$ denote the edges incident to $v$ according to the respective port numbers at $v$, and assume w.l.o.g.\ that $Q_i$ is the output at $(v, e_i)$.
		We start by assigning a value $\alpha(i) \in \{ \out, \ing \}$ to each index $1 \leq i \leq \Delta$ by setting $\alpha(i) := \ing$ if edge $e_i$ is oriented towards $v$, and $\alpha(i) := \out$ if edge $e_i$ is oriented away from $v$.
		Define values $\beta(i)$, by setting $\beta(i) := \alpha(i)$ if $Q_i \neq P_{\infty}$, and $\beta(i) := \none$ if $Q_i = P_{\infty}$.
		Let $R^v$ denote the multiset $\{ (Q_1, \beta(1)), \dots, (Q_{\Delta}, \beta({\Delta})) \}$.
		Let $H_1(\Delta)$ be the set of all possible $R^v$, i.e., $H_1(\Delta) := \{ \{ (Q'_1, \beta(1)), \dots, (Q'_{\Delta}, \beta({\Delta})) \} \mid Q' := \{ Q'_1, \dots, Q'_{\Delta} \} \in h_1(\Delta), \beta(1), \dots, \beta({\Delta}) \in \{ \out, \ing, \none \}, \beta(i) = \none \iff Q'_i = P^{Q'}_{\infty} \textrm{ for all } 1 \leq i \leq \Delta \}$.
		
		Recall that each output for the problem of superweak $k'$-coloring is a pair $(c,\gamma) \in \{ 1, \dots, k'\} \times \{ \rightarrow, \multimap, \bullet \}$.
		Let $(c_e, \gamma_e)$ denote the output $v$ will assign to $(v,e)$ according to $\AAA'$.
		For the first entry $c_e$, we find an arbitrary injective function $c : H_1(\Delta) \rightarrow \{ 1, \dots, k' \}$ in some fixed deterministic way (i.e., all nodes use the same function), and then set $c_e := c(R^v)$ for all $e$ incident to $v$.
		Since each element $Q' \in h_1(\Delta)$ is a multiset of cardinality $\Delta$ in which, by Lemma\ \ref{lem:infty}, all contained sets $\neq P^{Q'}_{\infty}$ together have multiplicity at most $2^{4^k}$, and since each contained set is a set of trit sequences of length $k$ (of which there are $2^{3^k}$ distinct ones), there are most $\left(3 \cdot 2^{3^k}\right)^{2^{4^k}+1} \leq k'$ distinct elements in $H_1(\Delta)$, by the definition of $H_1(\Delta)$.
		Hence, an injective function as described above exists.
		
		For the second entry $\gamma_e$, node $v$ computes a set $\JJ^*$ as described in Lemma\ \ref{lem:hall} (and whose existence is guaranteed by Lemma\ \ref{lem:hall}) in a fixed deterministic way.
		We require that the multiset $\{ (Q_i,\beta(i)) \mid i \in \JJ^* \}$ does not depend on the order of the pairs $(Q_i, \beta(i))$ fixed by the port numbers, i.e., that any two nodes $u, v$ with $R^u = R^v$ obtain the same \emph{multiset} $\{ (Q_i, \beta(i)) \mid i \in \JJ^* \}$ by their choice of $\JJ^*$.
		However, since whether a set $\JJ^*$ has the properties described in Lemma\ \ref{lem:hall} depends only on the multisets $\{ (Q_i, \alpha(i)) \mid i \in \JJ^* \} = \{ (Q_i, \beta(i)) \mid i \in \JJ^* \}$ and $\{ (Q_i, \alpha(i)) \mid i \in \{ 1, \dots, \Delta \}, Q_i \neq P_{\infty} \} = \{ (Q_i, \beta(i)) \mid i \in \{ 1, \dots, \Delta \}, Q_i \neq P_{\infty} \}$ (due\footnote{Note that, for the first equality, we use the fact that, for all $i \in \JJ^* \subseteq \II$, we have $Q_i \neq P_{\infty}$ since all indices $\ell \in \II$ satisfy $11\dots 1 \notin Q_{\ell}$ whereas $11\dots 1 \in P_{\infty}$, by Lemma\ \ref{lem:infty}.} to the definition of $\II$ in Lemma\ \ref{lem:hall}), such a choice for $\JJ^*$ exists.
		Now, set $\gamma_{e_i} := \hspace{0.3em} \rightarrow$ if $i \in \JJ^*$, $\gamma_{e_i} := \hspace{0.3em} \multimap$ if $i \in N(\JJ^*)$, and $\gamma_{e_i} := \bullet$ if $i \notin \JJ^* \cup N(\JJ^*)$.
		Note that $\JJ^* \cap N(\JJ^*) = \emptyset$, by Lemma\ \ref{lem:hall}.

		We argue that the algorithm $\AAA'$ specified by the above indeed solves superweak $k'$-coloring on $\GG$.
		Clearly $v$ outputs the same color $c$ at each $(v,e)$ according to $\AAA'$.
		Moreover, the number of demanding pointers $\rightarrow$ that $v$ outputs is strictly larger than the number of accepting pointers $\multimap$ since $|\JJ^*| > |N(\JJ^*)|$, by Lemma\ \ref{lem:hall}; also, by Lemma\ \ref{lem:infty}, the latter number is at most $2^{4^k} \leq k'$ since, for each $i \in N(\JJ^*)$, we have $Q_i \neq P_{\infty}$, by the definition of $\II$ in Lemma\ \ref{lem:hall} and the fact that $\JJ^* \subseteq \II$.
		What is left to show is that there is no conflict on an edge, i.e., that for each edge $e = \{ u, v \}$, $u$ and $v$ output different colors at $(u,e)$ and $(v,e)$, or both output $\bullet$, or at least one of them outputs $\multimap$.
		
		Hence, consider the case that $u$ and $v$ output the same color at $(u,e)$ and $(v,e)$, and that one of the two, say $u$, also outputs $\rightarrow$ at $(u,e)$.
		By the definition of $\AAA'$, this implies that the multisets $R^u$ and $R^v$ are equal (due to the injectivity of the above function), and that the set $Q_{u,e}$ that $\AAA$ outputs at $(u,e)$ satisfies $Q_{u,e} \in \{ Q_i \in Q^u \mid i \in \JJ^*_u \}$, where $\JJ^*_u$ is the index set $\JJ^*$ computed by $u$.
		Let $i$ be the port number for $e$ at $u$, and $j$ the port number for $e$ at $v$, and add the sub- or superscript $u$, resp.\ $v$, to the usual definitions to specify which node they correspond to.
		As $\{Q_{v,e}, Q_{u,e} \} \in g_1{\Delta}$ (by the correctness of $\AAA$) and $\alpha_u(i) \neq \alpha_v(j)$, we see that there is some index $\ell \in N_u(\JJ^*_u)$ with $Q^u_\ell = Q_{v,e} = Q^v_j$ and $\beta_u(\ell) = \beta_v(j)$, due to $R^u = R^v$.
		Since $\{ (Q^u_{r}, \beta_u(r)) \mid r \in \JJ^*_u \} = \{ Q^v_{r} \beta_v(r)) \mid r \in \JJ^*_v \}$ (due to $R^u = R^v$ and the specified way of choosing $\JJ^*$), there must be an index $r \in N_v(\JJ^*_v)$ such that $Q^v_r = Q^v_j$ and $\beta_v(r) = \beta_v(j)$.
		Observe that there cannot be two indices $j_1 \in N_v(\JJ^*_v), j_2 \notin N_v(\JJ^*_v)$ with $Q^v_{j_1} = Q^v_{j_2}$ and $\beta_v(j_1) = \beta_v(j_2)$, due to the definition of the function $N()$.
		Hence, we obtain $j \in N_v(\JJ^*_v)$, which implies that $v$ outputs $\multimap$ at $(v,e)$, by the definition of $\AAA'$.	
	\end{proof}
	
	Due to our speedup results in Section\ \ref{sec:speedup}, the following speedup lemma for superweak $k$-coloring is essentially a corollary of Lemma\ \ref{lem:zeroround}.
	As the setting for weak $2$-coloring used by Naor and Stockmeyer includes unique IDs, we formulate Lemma\ \ref{lem:supersuper} for order-invariant algorithms.
	An analogous speedup lemma can be achieved for general algorithms in a setting without unique IDs.
	
	\begin{lemma}\label{lem:supersuper}
		Let $k\geq 2$, and fix some $\Delta \geq 2^{4^k+1}$.
		Let $\GG_{n,\Delta}$ be a $\Sigma$-input-labeled graph class such that every contained graph has $n$ nodes, is $\Delta$-regular and the given input includes an orientation for each edge.
		Assume that $\GG_{n, \Delta}$ is $t$-independent for some positive integer $t$ and contains only graphs of girth at least $2t+2$.
		Let $S \subseteq \NN^+$ be a finite set of identifiers satisfying $|S| \geq 4\Delta^{2t}$.
		Let $\GG'_{n, \Delta}$ be the class consisting of all (input-labeled) graphs obtained by taking a graph from $\GG_{n, \Delta}$ and assigning unique identifiers from $S$ as inputs to the nodes of the graph.
		Set $k' := 2^{2^{5^k}}$.
		If there is an order-invariant algorithm solving superweak $k$-coloring on $\GG'_{n, \Delta}$ in time $t$, then there is also an order-invariant algorithm solving superweak $k'$-coloring on $\GG'_{n, \Delta}$ in time $t-1$.
	\end{lemma}
	
	\begin{proof}
		Set $\Pi$ to be superweak $k$-coloring and let $\AAA$ be an order-invariant algorithm solving $(\Pi, \GG'_{n, \Delta})$ in time $t$.
		Then, there is an order-invariant algorithm solving $(\Pi'_1, \GG'_{n, \Delta})$ in time $t-1$, by Theorem\ \ref{thm:ordinv}.
		It follows that there is also an algorithm $\AAA'$ solving superweak $k'$-coloring on $\GG'_{n, \Delta}$ in time $t-1$, by Lemma\ \ref{lem:zeroround}.
		The order-invariance of $\AAA'$ follows from the proof of Lemma\ \ref{lem:zeroround}, or, more precisely, from the fact that $\AAA'$, as defined in that proof, only takes the outputs of (the order-invariant algorithm) $\AAA$ and the initial input information into account.
	\end{proof}

	\subsection{Proving the Lower Bound}\label{sec:lowerbound}
	
	In this section, we show that there is no algorithm solving weak $2$-coloring in time $o(\logstar \Delta)$ in odd-degree graphs in the setting used by Naor and Stockmeyer \cite{DBLP:conf/stoc/NaorS93}, or in the $\LOCAL$ model \cite{DBLP:journals/siamcomp/Linial92, Peleg2000}.
	Those two models differ from the port numbering model we use in that they provide each node with a globally unique ID, where the $\LOCAL$ model (commonly) uses IDs that are $O(\log n)$-bit strings while Naor and Stockmeyer assume arbitrarily large IDs.
	As the setting with a bound on the size of the IDs, i.e., the $\LOCAL$ model, clearly makes proving a lower bound harder, we will formally prove our lower bound for the $\LOCAL$ model.

	As mentioned in Section\ \ref{sec:lift}, Naor and Stockmeyer themselves provide a tool to circumvent the complications arising from unique IDs: as they show in their work, if there is a constant-time algorithm for some problem, then there is also an order-invariant algorithm for the same problem with the same runtime.
	We will make use of this fact in the proof of Theorem\ \ref{thm:lsd}.

	\begin{theorem}\label{thm:lsd}
		There is no $o(\logstar \Delta)$-time algorithm solving weak $2$-coloring in odd-degree graphs.
	\end{theorem}

	\begin{proof}
		Assume for a contradiction that such an $o(\logstar \Delta)$-time algorithm $\AAA$ exists, and, for each odd $\Delta \in \NN$, let $T(\Delta) \in o(\logstar \Delta)$ denote the (worst-case) runtime of $\AAA$ on $\Delta$-regular graphs.
		We can assume w.l.o.g.\ that $T(\Delta) \geq 1$ for all $\Delta \in \NN$ (by choosing our algorithm $\AAA$ suitably).
		Fix $\Delta > 16$ to be an odd positive integer such that $1 \leq T(\Delta) \leq (\logstar \Delta - 7) / 5$ (such a $\Delta$ must exist since $T(\Delta) \in o(\logstar \Delta)$).

		Let $\GG_{n, \Delta}$ be the ($\Sigma$-input-labeled) graph class consisting of all $\Delta$-regular graphs with $n$ nodes, girth at least $2(T(\Delta)+1) + 2$, and arbitrary edge orientations as inputs, where $n$ is a sufficiently large even constant (in particular, the ID space, which depends on $n$, has to be of cardinality at least $4 \Delta^{2 T(\Delta)+2}$ since we want to apply Lemma\ \ref{lem:supersuper}).
		Let $\GG'_{n, \Delta}$ be the class consisting of all graphs obtained by taking a graph from $\GG_{n, \Delta}$ and assigning unique $O(\log n)$-bit IDs as inputs to the nodes of the graph.
		Let $\GG''_{n, \Delta}$ be the graph class obtained from $\GG'_{n, \Delta}$ by removing the edge orientations from the input (as those are not part of the $\LOCAL$ model).
		We will show that our lower bound already holds on $\GG''_{n, \Delta}$.
		Note that the three defined graph classes are all non-empty as assured, e.g., by \cite[Chapter III, Theorem $1.4'$]{bollobas78extremal}.

		Since, by assumption, $\AAA$ solves weak $2$-coloring on $\GG''_{n, \Delta}$ in time $T(\Delta) \leq (\logstar \Delta - 8) / 5$, and since $\Delta$ is a fixed constant, there must also be an \emph{order-invariant} algorithm doing the same, by \cite[Theorem 3.3]{DBLP:conf/stoc/NaorS93} (or, more precisely, by the proof thereof).
		Hence, in the following assume that $\AAA$ is order-invariant.
		Clearly, our order-invariant algorithm $\AAA$ also solves weak $2$-coloring on $\GG'_{n, \Delta}$ in time $T(\Delta) \leq (\logstar \Delta - 8) / 5$, by simply ignoring the additional input information.
		As argued in Section\ \ref{sec:app2col}, any algorithm that solves weak $2$-coloring can be transformed into an algorithm that solves the pointer version of weak $2$-coloring and requires at most one additional round.
		It is straightforward to check that this transformation also preserves order-invariance; for simplicity, denote the algorithm obtained after the transformation also by $\AAA$.
		Moreover, any algorithm that solves the pointer version of weak $2$-coloring also solves superweak $2$-coloring, by the definition of the latter.
		Hence, it follows that $\AAA$ solves superweak $2$-coloring on $\GG'_{n, \Delta}$ in time $T(\Delta) + 1 \leq (\logstar \Delta - 3) / 5$.
		
		Observe that the definition of our graph class $\GG_{n, \Delta}$ ensures that $\GG_{n, \Delta}$ is $(\leq T(\Delta)+1)$-independent.
		Hence we can use Lemma\ \ref{lem:supersuper} to repeatedly speed $\AAA$ up until we obtain an algorithm that solves superweak $k'$-coloring for some large $k'$ in $0$ rounds.
		Set $k_0 := 2$, and define recursively $k_{i+1} := F(F(F(F(F(k_i)))))$, where $F(x) := 2^x$.
		Since $k_{i+1} \geq 2^{2^{5^{k_i}}}$ for all $i$ and any algorithm for superweak $j$-coloring also solves superweak $j'$-coloring if $j' \geq j$, we see that $T(\Delta)+1$ applications of Lemma\ \ref{lem:supersuper} result in a $0$-round algorithm $\AAA^*$ that solves superweak $k^*$-coloring on $\GG'_{n, \Delta}$, where $k^* \leq k_{T(\Delta)+1} \leq \log \Delta$.
		Note that $k_i \leq \log \log \log \log \Delta$ for all $i \leq T(\Delta)$, which implies that the condition $\Delta \geq 2^{4^k+1}$ in Lemma\ \ref{lem:supersuper} is satisfied in each application of the lemma.
		In the following we argue that such a $0$-round algorithm cannot exist.
		
		Consider a node $v$ in some graph from $\GG'_{n, \Delta}$, and assume that $v$ has only incoming edges at the first $(\Delta-1)/2$ ports and only outgoing edges at the remaining $(\Delta+1)/2$ ports.
		Let $\ident_1$ and $\ident_2$ be two IDs such that $v$ would output the same color according to $\AAA^*$ if it has either of the two IDs as input.
		Such two IDs must exist by the pidgeonhole principle.
		Now assume that $v$ is given $\ident_1$ as input, and let $e=\{u, v\}$ be some edge incident to $v$ such that $v$ outputs $\rightarrow$ at $(v,e)$ according to $\AAA^*$.
		Assume that $u$ has the same constraint on the incoming and outgoing edges as $v$ (which is possible due to our choice of $\GG'_{n, \Delta}$) and let $\ident_2$ be the input ID for $u$, which implies that $u$ outputs the same color as $v$.
		Moreover, since $k^* \leq \log \Delta \leq (\Delta-3)/2$ (due to $\Delta > 16$), there are at least one edge $e_{\ing}$ incoming at $u$ and one edge $e_{\out}$ outgoing at $u$ such that $u$ outputs $\multimap$ neither at $(u, e_{\ing})$ nor at $(u, e_{\out})$ (since $u$ has at most $k^*$ pointers $\multimap$ available).
		Clearly, there exists a port numbering at $u$ such that $e_{\ing} = e$ or $e_{\out} = e$.
		In this case, $u$ and $v$ output the same color, $v$ outputs $\rightarrow$ at $(v,e)$, and $u$ does not output $\multimap$ at $(u,e)$, which yields an incorrect output for superweak $k^*$-coloring at $e$.
		Hence, there is no $o(\logstar \Delta)$-time algorithm as described in the theorem statement.	
	\end{proof}

	\section{Conclusion}
	In this work, we developed a new technique for determining (or at least bounding) the time complexity of a given distributed problem, based on the idea of automatically transforming the problem into a problem that can be solved exactly one round faster.
	We proved the viability of the technique by showing that it can be used for reproducing known results in a (semi-)automatic fashion and that it is powerful enough to facilitate answering a long-standing open question that has resisted proof attempts for 25 years.
	Given that the technique can be applied to any locally checkable problem, we expect many other problems to be solved by this technique.
	A first confirmation has already been given by the follow-up work of Balliu et al.\ \cite{DBLP:journals/corr/abs-1901-02441}, where the authors use our speedup technique to prove lower bounds for maximal matching and maximal independent set.
	
	One main difficulty in applying our speedup technique lies in the fact that the complexity of the problem description increases substantially in each speedup step.
	We provided a simplification technique (maximization) and outlined two general approaches (relaxation for a lower bound, making a problem harder for an upper bound) that mitigate this problem.
	Are there other techniques for reducing the description complexity?
	Is there a good way to determine which parts of the description of a derived problem are important and which can be discarded?
	Is finding the correct lower bound only a matter of finding the right relaxation (and to which degree can this be automated)?
	Answering these questions would be an important step forward in establishing the presented speedup as a main technique for determining the time complexity of locally checkable problems.
	
\section*{Acknowledgments}
I would like to thank Alkida Balliu, Dennis Olivetti, and Jukka Suomela for insghtful discussions about weak coloring.

	\bibliographystyle{plain}
	\bibliography{references}
	
\end{document}